\documentclass[10pt,twocolumn]{IEEEtran}

\usepackage{amsfonts}
\usepackage[latin2]{inputenc}
\usepackage{t1enc}
\usepackage{amssymb}
\usepackage{amsmath}
\usepackage[mathscr]{eucal}
\usepackage{amsthm}
\usepackage{array}
\usepackage{color}
\usepackage{cite}

\theoremstyle{definition}
\newtheorem{defin}{Definition}[section]
\newtheorem{thm}[defin]{Theorem}
\newtheorem{prop}[defin]{Proposition}
\newtheorem{rem}[defin]{Remark}

\newtheorem{cor}[defin]{Corollary}

\newtheorem{Lemma}[defin]{Lemma}

\def\hil{{\mathcal H}}
\def\kil{{\mathcal K}}
\def\B{{\mathcal B}}
\def\C{{\mathcal C}}
\def\D{{\mathcal D}}
\def\I{{\mathcal I}}

\def\N{{\mathcal N}}
\def\P{{\mathcal P}}
\def\S{{\mathcal S}}
\def\W{{\mathcal W}}
\def\X{{\mathcal X}}

\def\bz{\left(}
\def\jz{\right)}

\def\ki{\textit}
\def\kii{\textit}
\def\kiii{}
\def\vfi{\varphi}
\def\half{\frac{1}{2}}
\def\ep{\varepsilon}

\def\bN{\mathbb{N}}

\def\bR{\mathbb{R}}
\def\bC{\mathbb{C}}
\def\prod{\mathrm{prod}}

\def\old{}
\def\oldd{\{\s\}}
\def\nw{^{*}}
\def\nww{\ast}
\def\x{^{(t)}}
\def\xx{(t)}
\def\typ{(t)}
\def\inv{^{-1}}
\def\what{\mathbb}
\def\map{\Phi}
\def\notimes{^{(\otimes n)}}
\def\derleft{\partial^{-}}

\def\ol{\overline}
\def\sa{_\mathrm{sa}}

\newcommand{\norm}[1]{\left\| #1\right\|}

\newcommand{\rsr}[3]{D_{#3}\bz #1\|#2\jz}
\newcommand{\rsro}[3]{D_{#3}\bz #1\|#2\jz}
\newcommand{\rsrn}[3]{D^{*}_{#3}\bz #1\|#2\jz}

\newcommand{\bigrsrn}[3]{D\nw_{#3}\bz #1\Big\|#2\jz}

\newcommand{\s}{\mbox{ }}
\newcommand{\ds}{\mbox{ }\mbox{ }}

\newcommand{\diad}[2]{|#1\rangle\langle #2|}
\newcommand{\pr}[1]{\diad{#1}{#1}}

\newcommand{\vecc}[1]{\underline{#1}}
\newcommand{\hol}[1]{\chi_{_{#1}}}

\DeclareMathOperator{\Tr}{Tr}

\DeclareMathOperator{\ran}{ran}
\DeclareMathOperator{\supp}{supp}

\DeclareMathOperator{\id}{id}
\DeclareMathOperator{\logn}{\widehat\log}

\begin{document}

\title{Coding theorems for compound problems via quantum R\'enyi divergences}
\author{Mil\'an Mosonyi \thanks{M.~Mosonyi is with the F\'{\i}sica Te\`{o}rica: Informaci\'{o} i Fenomens Qu\`{a}ntics,
Universitat Aut\`{o}noma de Barcelona, ES-08193 Bellaterra (Barcelona), Spain, and with the Mathematical Institute, Budapest University of Technology and Economics,
Egry J\'ozsef u~1., Budapest, 1111 Hungary (e-mail: milan.mosonyi@gmail.com)}\\
\thanks{This paper was presented in part at the Conference on the Theory of Quantum Computation, Communication \& Complexity (TQC 2014), 21--23 May 2014, Singapore.}}

\maketitle

\begin{abstract}
Recently, a new notion of quantum R\'enyi divergences has been introduced by 
M\"uller-Lennert, Dupuis, Szehr, Fehr and Tomamichel, 
J.~Math.~Phys. \textbf{54}:122203, (2013), and Wilde, Winter, Yang, 
Commun.~Math.~Phys. \textbf{331}:593--622, (2014),
that has found a number of applications in strong converse theorems. 
Here we show that these new R\'enyi divergences are also 
useful tools to obtain coding theorems in the direct domain of various problems. 
We demonstrate this by giving new and considerably simplified proofs for the achievability
parts of Stein's lemma with composite null hypothesis, universal state compression, 
and the classical capacity of compound classical-quantum channels, based on
single-shot error bounds already available in the literature, and simple properties
of the quantum R\'enyi divergences. The novelty of our proofs is that 
the composite/compound coding theorems can be almost directly obtained from the single-shot error bounds,
with essentially the same effort as for the case of
simple null-hypothesis/single source/single channel.
\end{abstract}

\section{Introduction}

R\'enyi introduced a generalization 
of the Kullback-Leibler divergence (relative entropy) in \cite{Renyi}. According to his 
definition, the $\alpha$-divergence of two probability distributions 
$p$ and $q$ on a finite set $\X$ for a parameter $\alpha\in[0,+\infty)\setminus\{1\}$ is 
given by
\begin{align}\label{Renyi def}
\rsr{p}{q}{\alpha}:=
\frac{1}{\alpha-1}\log\sum_{x\in\X}p(x)^{\alpha}q(x)^{1-\alpha}.
\end{align}
The limit $\alpha\to 1$ yields the standard relative entropy.
These quantities turned out to play a central role in information theory and statistics; 
indeed, the R\'enyi divergences quantify the trade-off between the exponents of the relevant 
quantities in many information-theoretic tasks, including hypothesis testing, source coding 
and noisy channel coding; see, e.g.~\cite{Csiszar} for an overview of these results. It was 
also shown in \cite{Csiszar} that the R\'enyi divergences, and other related 
quantities, like the R\'enyi entropies and the R\'enyi capacities, have direct operational 
interpretations as so-called generalized cutoff rates in the corresponding 
information-theoretic tasks.

In quantum theory, the state of a system is described by a density operator instead of a 
probability distribution, and the definition \eqref{Renyi def} can be extended for pairs of 
density operators in various inequivalent ways, due to 
the non-commutativity of operators.
The traditional way to define the R\'enyi divergence of two density operators is
\begin{equation}\label{old Renyi}
\rsro{\rho}{\sigma}{\alpha}:=
\frac{1}{\alpha-1}\log\Tr\rho^{\alpha}\sigma^{1-\alpha}.
\end{equation}
The quantum Hoeffding bound theorem \cite{ANSzV,Hayashi,HMO2,Nagaoka} shows that these divergences, with $\alpha\in(0,1)$,
play the same role in quantifying the trade-off of the two error probabilities in the direct domain of binary state disrcimination as their classical counterparts \eqref{Renyi def} in classical hypothesis testing.
Based on the Hoeffding bound theorem, a direct operational interpretation of these divergences has been given in \cite{MH}.

Recently, a new quantum extension of the R\'enyi $\alpha$-divergences has been proposed in \cite{Renyi_new,WWY}, defined as
\begin{equation}\label{new Renyi def}
\rsrn{\rho}{\sigma}{\alpha}:=
\frac{1}{\alpha-1}\log\Tr\bz\sigma^{\frac{1-\alpha}{2\alpha}}\rho\sigma^{\frac{1-\alpha}{2\alpha}}\jz^{\alpha}.
\end{equation}
This definition was introduced in \cite{Renyi_new} as a parametric family that connects
the min- and max-relative entropies \cite{Datta,RennerPhD} and Umegaki's relative entropy
\cite{Umegaki}. In \cite{WWY}, the corresponding generalized Holevo capacities were used 
to establish the strong converse property for the classical capacities of 
entanglement-breaking and Hadamard channels.
It was shown in \cite{MO} that these new R\'enyi divergences play the same role in the 
(strong) converse problem of binary state discrimination as the traditional R\'enyi divergences in the direct problem. In particular, the strong converse exponent was expressed as a function of the new R\'enyi divergences, and from that a direct operational interpretation was derived for $D_{\alpha}\nw$, $\alpha>1$, as generalized cutoff rates in the sense of 
\cite{Csiszar}. 
Exact strong converse exponents in terms of quantities derived from $D_{\alpha}^*$ have since been obtained for
other types of discrimination problems \cite{CMW,HT,MO2},
as well as for classical-quantum channel coding \cite{MO3}

So far, it seems that the new quantum R\'enyi divergences $D_{\alpha}^*$ find their application in
strong converse theorems, and for the parameter range $\alpha>1$, while the natural quantities for 
the direct part of coding theorems are the traditional $D_{\alpha}$ quantities, with parameters 
$\alpha\in(0,1)$. Our aim here is to show that the new R\'enyi divergences, and with parameters $\alpha\in(0,1)$, are also useful 
to obtain the direct parts of various coding theorems. We demonstrate this by giving new proofs for the achievability parts
of the quantum Stein's lemma with composite null hypothesis \cite{BDKSSSz,Notzel}, universal state compression \cite{JHHH}, and the classical capacity of compound classical-quantum channels \cite{BB,DD}.
We will follow the following unified approach to these coding theorems:
\begin{enumerate}
\item
We start with a single-shot coding theorem that bounds the relevant error probability in terms of a R\'enyi divergence. 
In the case of Stein's Lemma and source compression, this will be Audenaert's inequality \cite{Aud}, while in the case of 
channel coding, we use the random coding theorem due to Hayashi and Nagaoka \cite{HN}. The bounds in both cases are 
in terms of $Q_{\alpha}\old=\exp((\alpha-1)D_{\alpha}\old)$; for instance, in the case of state discrimination, the divergence term of the bound is of the form
$Q_{\alpha}\old(\sum_{\rho}\rho\|\sigma)$, where the summation is over the elements of the composite null-hypothesis set, and $\sigma$ is the alternative hypothesis.

\item
We use the Araki-Lieb-Thirring inequality to further upper bound the $Q_{\alpha}$ term by 
$Q_{\alpha}\nw=\exp\bz(\alpha-1)D_{\alpha}\nw\jz$.
The purpose of this is to benefit from a simple subadditivity property of  $Q_{\alpha}^*$, that allows to decouple
the upper bound into a sum of pairwise terms, e.g., $Q_{\alpha}\nw(\sum_{\rho}\rho\|\sigma)$ 
into $\sum_{\rho}Q_{\alpha}\nw(\rho\|\sigma)$
in the above example.

\item
We may also use a converse to the Araki-Lieb-Thirring inequality, due do Audenaert \cite{Aud-ALT}, to convert the
$D_{\alpha}\nw$ divergences back to $D_{\alpha}$, if that offers a simplification of the proof.

\item
Finally, we apply the above bounds to many copies, and take the number of copies to infinity.
\end{enumerate}

The advantage of the above approach is that it only uses very 
general arguments that are largely independent of the concrete model in consideration.
Once the single-shot coding theorems are available, the coding theorems for the composite/compound cases follow 
essentially by the same 
amount of effort as for the simple cases (simple null-hypothesis, single source, single channel), using only very general 
properties of the R\'enyi divergences. This makes the proofs considerably shorter and simpler than e.g., in 
\cite{BDKSSSz,BB,DD}. Moreover, this approach is very easy to generalize to non-i.i.d.~compound problems, 
as it does not rely on the method of types, cf.~\cite{JHHH,Notzel}.

We would also like to emphasize the technical simplicity of the proofs; the only technically more involved ingredients are
the Araki-Lieb-Thirring inequality \cite{Araki,LT} and its converse \cite{Aud-ALT}, and the Hayashi-Nagaoka random coding 
lemma \cite{HN}.
\smallskip

The structure of the paper is as follows. In Section \ref{sec:notations} we collect the necessary preliminaries.
In Section \ref{sec:Renyi}, we review some properties of the R\'enyi divergences and the related notion of $\alpha$-capacities. 
The new contribution towards the study of R\'enyi divergences
are the lower bounds in Lemma \ref{Lemma:old-new} and Proposition \ref{prop:complements}, both of which we will utilize 
in the coding theorems in Section \ref{sec:applications}, together with other technical lemmas,
Lemma \ref{Lemma:dist limit} and Lemma \ref{Lemma:inf Hol limit}. 
Since the new type of R\'enyi divergences have been introduced very recently, and their properties and applications
are at the moment being intensively explored in the literature, we also include some observations
in Section \ref{sec:Renyi} that are not directly necessary for Section \ref{sec:applications}.
This is partly to put other things into a broader context (e.g., connecting Proposition \ref{prop:complements}
to the very important convexity properties of the R\'enyi quantities in Section \ref{sec:convexity}),
and partly in the hope of possible future applications (e.g., for Remark \ref{rem:quant} and 
Lemma \ref{Lemma:chi bounds}).

The main contribution of the paper is Section \ref{sec:applications}, where we prove the achievability parts of Stein's lemma with composite null-hypothesis in Section \ref{sec:Stein}, for universal state compression in Section \ref{sec:comp},
and for classical-quantum channel coding in Section \ref{sec:capacity}, following the approach outlined above.

\section{Preliminaries}
\label{sec:notations}

For a finite-dimensional Hilbert space $\hil$, let $\B(\hil)_+$ denote the set of all 
non-zero positive semidefinite operators on $\hil$, and let 
$\S(\hil):=\{\rho\in\B(\hil)_+:\,\Tr\rho=1\}$ be the set of all \ki{density operators (states)}
on $\hil$. We use the notation $\B(\hil)\sa$ for the set of self-adjoint operators on $\hil$.

We define the powers of a positive semidefinite operator $A$ only on its support; that is, 
if $\lambda_1,\ldots,\lambda_r$ are the strictly positive eigenvalues of $A$, with corresponding 
spectral projections $P_1,\ldots,P_r$, then we define
$A^{\alpha}:=\sum_{i=1}^r \lambda_i^{\alpha}P_i$ for all $\alpha\in\bR$. In particular, 
$A^0=\sum_{i=1}^rP_i$ is the projection onto the support of $A$. 

For a self-adjoint operator $X$, we will use the notation $\{X>0\}$ to denote the spectral projection of $X$ corresponding 
to the positive half-line $(0,+\infty)$. The spectral projections $\{X\ge 0\},\,\{X<0\}$ and $\{X\le 0\}$ are defined similarly.
The positive part $X_+$ and the negative part $X_-$ are defined as $X_+:=X\{X>0\}$ and $X_-:=-X\{X<0\}$, respectively, and the absolute value of $X$ is $|X|:=X_++X_-$. The \ki{trace-norm} of $X$ is $\norm{X}_1:=\Tr|X|$.

The following Lemma is Theorem 1 from \cite{Aud}; see also Proposition 1.1 in \cite{JOPS} for a simplified proof.
\begin{Lemma}\label{Lemma:Aud}
Let $A,B$ be positive semidefinite operators on a Hilbert space. For any $t\in[0,1]$,
\begin{align*}
&\Tr A(I-\{A-B>0\})+\Tr B\{A-B>0\}\\
&\ds=\half\Tr(A+B)-\half\norm{A-B}_1\\
&\ds\le\Tr A^tB^{1-t}.
\end{align*}
\end{Lemma}
\medskip

The closeness of two operators can be measured in various ways. Apart from the trace-norm, we will also use the 
\ki{operator norm}, defined for an operator $A\in\B(\hil)$ as $\norm{A}:=\max\{\norm{Ax}:\,x\in\hil,\,\norm{x}\le 1\}$.
The \ki{fidelity} of positive semidefinite operators $A$ and $B$ is defined as 
$F(A,B):=\Tr\bz A^{1/2}BA^{1/2}\jz^{1/2}$.

The \ki{entanglement fidelity} of a state $\rho$ and a completely positive trace-preserving (CPTP) map $\map$ is
$F_e(\rho,\map):=F\bz\pr{\psi_{\rho}},(\id\otimes \map)\pr{\psi_{\rho}}\jz$, where $\psi_{\rho}$ is any purification of the state $\rho$; see Chapter 9 in \cite{NC}
for details.
\smallskip

The next Lemma is a reformulation of Lemma 2.6 in \cite{MS}. We include the proof for readers' convenience.
\begin{Lemma}\label{Lemma:state approximation}
Let $(V,\norm{.})$ be a finite-dimensional real or complex  normed vector space, and let $\dim_{\bR}V$ denote its real dimension.
Let $\N$ be a subset of the unit ball of $V$. For every $\delta>0$, there exists a finite subset $\N_{\delta}\subset\N$ such that 
\smallskip

\noindent 1. $\displaystyle{|\N_{\delta}|\le (1+2/\delta)^{\dim_{\bR}V}}$, 
and
\smallskip

\noindent 2. for every $v\in\N$ there exists a $v_{\delta}\in\N_{\delta}$ such that
$\norm{v-v_{\delta}}<\delta$.
\end{Lemma}
\begin{proof}
For every $\delta>0$, let $\N_{\delta}$ be a maximal set in $\N$ such that $\norm{v-v'}\ge \delta$ for every
$v,v'\in\N_{\delta}$; then $\N_{\delta}$ clearly satisfies 2.
On the other hand, the open $\norm{\s}$-balls with radius 
$\delta/2$ around the elements of $\N_{\delta}$ are disjoint, and contained in the $\norm{\s}$-ball with radius $1+\delta/2$
and origin $0$. Since the volume of balls scales with their radius on the power $\dim_{\bR}V$, we obtain 1.
\end{proof}
\medskip

The following minimax theorem is Corollary A.2 in \cite{MH}:
\begin{Lemma}\label{Lemma:minimax}
Let $X$ be a compact topological space, $Y$ be a subset of the real line, and 
$f:\,X\times Y\to\bR\cup\{-\infty,+\infty\}$ be such that for every $y\in Y$, $f(.,y)$ is lower semicontinuous on $X$,
and for every $x\in X$, $f(x,.)$ is monotone increasing on $Y$. Then
\begin{align*}
\inf_{x\in X}\sup_{y\in Y}f(x,y)=\sup_{y\in Y}\inf_{x\in X}f(x,y),
\end{align*}
and the infima can be replaced with minima.
\end{Lemma}

For the natural logarithm function $\log$, we will use the convention 
\begin{align*}
\log 0:=-\infty\ds\ds\text{and}\ds\ds\log +\infty:=+\infty.
\end{align*}
We also introduce the notation
\begin{align}\label{sign}
s(\alpha):=\begin{cases}
1,&\alpha\in[0,1],\\
-1,&\alpha>1.
\end{cases}
\end{align}

\section{R\'enyi divergences}
\label{sec:Renyi}

\subsection{Two definitions}\label{sec:ALT}

For non-zero positive semidefinite operators $\rho,\sigma$, 
and every $\alpha\in(0,+\infty)$, let
\begin{align}\label{Q def}
Q_{\alpha}\old(\rho\|\sigma)&:=\Tr\rho^{\alpha}\sigma^{1-\alpha},\nonumber\\
Q_{\alpha}\nw(\rho\|\sigma)&:=\Tr\bz\sigma^{\frac{1-\alpha}{2\alpha}}\rho\sigma^{\frac{1-\alpha}{2\alpha}}\jz^{\alpha},
\end{align}
and define
\begin{align*}
\psi_{\alpha}\x(\rho\|\sigma):=\log Q_{\alpha}\x(\rho\|\sigma),\ds\ds\ds\typ=\oldd\ds\text{or}\ds\typ=\nww.
\end{align*}
Here and henceforth $\oldd$ stands for the empty string, i.e., $Q_{\alpha}\x$ with $\typ=\oldd$ is simply $Q_{\alpha}$.
For positive definite operators $\rho,\sigma$, 
the \ki{R\'enyi $\alpha$-divergences} \cite{Renyi} of $\rho$ w.r.t.~$\sigma$ 
with parameter $\alpha\in(0,+\infty)\setminus\{1\}$ are defined as
\begin{align}
D_{\alpha}\x(\rho\|\sigma)&:=\frac{1}{\alpha-1}\log Q_{\alpha}\x(\rho\|\sigma)-\frac{1}{\alpha-1}\log\Tr\rho\nonumber\\
&=
\frac{\psi_{\alpha}\x(\rho\|\sigma)-\psi_{1}\x(\rho\|\sigma)}{\alpha-1}.\label{D def inv}
\end{align}
For not necessarily invertible operators the definition is extended by 
\begin{align}\label{def by limit}
D_{\alpha}\x(\rho\|\sigma):=\lim_{\ep\searrow 0}D_{\alpha}\x(\rho+\ep I\|\sigma+\ep I).
\end{align}
It is easy to see that these limits exist, and we get
\begin{align*}
\rsro{\rho}{\sigma}{\alpha}&=
\frac{1}{\alpha-1}\log\Tr\rho^{\alpha}\sigma^{1-\alpha}-\frac{1}{\alpha-1}\log\Tr\rho,\\
\rsrn{\rho}{\sigma}{\alpha}&=
\frac{1}{\alpha-1}\log\Tr\bz\sigma^{\frac{1-\alpha}{2\alpha}}\rho\sigma^{\frac{1-\alpha}{2\alpha}}\jz^{\alpha}
-\frac{1}{\alpha-1}\log\Tr\rho
\end{align*}
when $\alpha\in(0,1)$ or $\supp\rho\subseteq\supp\sigma$, and $D_{\alpha}\x(\rho\|\sigma)=+\infty$ otherwise.

$Q_{\alpha}$ is a so-called \ki{quasi-entropy} or \ki{quantum $f$-divergence}, corresponding to the power function 
$x^{\alpha}$ \cite{HMPB,Petz}; its convexity and monotonicity properties \cite{Ando,Lieb,MH,Petz,HMPB} are of central importance 
for quantum information theory \cite{LR,NC,Petzbook,Wildebook}. The corresponding R\'enyi divergence $D_{\alpha}$ has been used 
in quantum information theory for a long time \cite{Hayashibook,Nagaoka2,ON,ON2} in bounds on the error probability in various 
information-theoretic tasks, and it has been shown recently to have a direct operational interpretation for $\alpha\in(0,1)$
in the problem of the \ki{quantum Hoeffding bound} \cite{Aud,ANSzV,Hayashi,Nagaoka}.
The R\'enyi divergence $D_{\alpha}\nw$ has been introduced recently in \cite{Renyi_new,WWY}, and has found applications in
various strong converse problems since then \cite{CMW,MO,MO2,WWY}.

\begin{rem}
It is easy to see that for non-zero $\rho$, we have $\lim_{\sigma\to 0}\rsro{\rho}{\sigma}{\alpha}=\lim_{\sigma\to 0}\rsrn{\rho}{\sigma}{\alpha}=+\infty$, and hence we define 
$\rsro{\rho}{0}{\alpha}:=\rsrn{\rho}{0}{\alpha}:=+\infty$ when $\rho\ne 0$. On the other hand, for non-zero $\sigma$, the limits
$\lim_{\rho\to 0}\rsro{\rho}{\sigma}{\alpha}$ and $\lim_{\rho\to 0}\rsrn{\rho}{\sigma}{\alpha}$ don't exist, and hence we don't define the values of 
$\rsro{0}{\sigma}{\alpha}$ and $\rsrn{0}{\sigma}{\alpha}$. Indeed, one can consider $\rho_n:=\frac{1}{n}\pr{0}+\frac{1}{n^{\beta}}\pr{1}$, and 
$\sigma:=\pr{1}$, where $\pr{0}$ and $\pr{1}$ are orthogonal rank $1$ projections. It is easy to see that for $\alpha<1$,
$\lim_{n\to +\infty}\rsro{\rho_n}{\sigma}{\alpha}= 
\lim_{n\to +\infty}\rsrn{\rho_n}{\sigma}{\alpha}=
\lim_{n\to +\infty}\frac{1}{\alpha-1}\log\frac{n^{1-\beta\alpha}}{1+n^{1-\beta}}$ depends on the value of $\beta$. A similar example can be used for $\alpha>1$.
\end{rem}

For invertible $\rho$ and $\sigma$, the second derivative of $\alpha\mapsto \psi_{\alpha}(\rho\|\sigma)$
is easily seen to be non-negative, and hence, by \eqref{D def inv}, 
\begin{align}\label{mon}
\alpha\mapsto D_{\alpha}\old(\rho\|\sigma)\ds\ds\text{ is monotone increasing}.
\end{align}
The same holds for general $\rho$ and $\sigma$ due to \eqref{def by limit}.
As a consequence, the \ki{R\'enyi entropies} 
\begin{align*}
S_{\alpha}(\rho)&:=-\rsro{\rho}{I}{\alpha}=-\rsrn{\rho}{I}{\alpha}\\
&=\frac{1}{1-\alpha}\log
\Tr\rho^{\alpha}-\frac{1}{1-\alpha}\log\Tr\rho
\end{align*}
are monotonic decreasing in $\alpha$ for any 
fixed $\rho$, and hence
\begin{align}\label{power bound}
s(\alpha)\Tr\rho^{\alpha}\le s(\alpha)(\Tr\rho^0)^{(1-\alpha)}(\Tr\rho)^{\alpha},\ds\ds\ds\alpha\in(0,+\infty).
\end{align}

It is straightforward to verify that 
$D_{\alpha}\old$ yields Umegaki's \ki{relative entropy} \cite{Umegaki,Wehrl} in the limit 
$\alpha\to 1$; i.e., for any $\rho,\sigma\in\B(\hil)_+$,
\begin{align}\label{1 limit}
\rsr{\rho}{\sigma}{1}&:=\lim_{\alpha\to 1}\rsro{\rho}{\sigma}{\alpha}\nonumber\\
&=
\begin{cases}
\frac{1}{\Tr\rho}\Tr\rho(\logn\rho-\logn\sigma),&\supp\rho\subseteq\supp\sigma,\\
+\infty,&\text{otherwise}.
\end{cases}
\end{align}
In the above formula, $\logn X$ stands for the logarithm of $X\in\B(\hil)_+$ taken on its support, 
and defined to be $0$ on the orthocomplement of its support.
The same limit relation has been shown to hold for $D_{\alpha}^*$ in \cite{Renyi_new}, and in \cite{WWY} for 
$\alpha\searrow 1$, by explicitly computing the derivative of 
$\alpha\mapsto\psi_{\alpha}\nw(\rho\|\sigma)$
at $\alpha=1$. We give an alternative derivation in Corollary \ref{cor:1 limit}.
\smallskip

It has been noted in \cite{WWY} that the Araki-Lieb-Thirring inequality \cite{Araki,LT} yields the ordering
$\rsrn{\rho}{\sigma}{\alpha}\le\rsro{\rho}{\sigma}{\alpha}$. The inequalities in 
\eqref{old-new bounds0}--\eqref{old-new bounds7} below complement this inequality.

\begin{Lemma}\label{Lemma:old-new}
For any $\rho,\sigma\in\B(\hil)_+$, and any $\alpha\in(0,+\infty)$,
\begin{align}
\rsro{\rho}{\sigma}{\alpha}
\ge&
\rsrn{\rho}{\sigma}{\alpha}\nonumber\\
\ge&
\alpha\rsro{\rho}{\sigma}{\alpha}+\log\Tr\rho
-\log\Tr\rho^{\alpha}\nonumber\\
&+(\alpha-1)\log\norm{\sigma}.\label{old-new bounds}
\end{align}
If $\rho$ is a density operator then 
\begin{align}
\label{old-new bounds9}
\rsro{\rho}{\sigma}{\alpha}
\ge&
\rsrn{\rho}{\sigma}{\alpha}\nonumber\\
\ge&
\alpha\rsro{\rho}{\sigma}{\alpha}-|\alpha-1|\max\{0,1-\alpha\}\log\Tr\rho^0\nonumber\\
&+(\alpha-1)\log\norm{\sigma},
\end{align}
and if also $\sigma$ is a density operator then
\begin{align}
\rsro{\rho}{\sigma}{\alpha}
\ge&
\rsrn{\rho}{\sigma}{\alpha}\nonumber\\
\ge&
\alpha\rsro{\rho}{\sigma}{\alpha}-|\alpha-1|\log\max\{\Tr\rho^0,\Tr\sigma^0\}\label{old-new bounds7}\\
\ge&
\alpha\rsro{\rho}{\sigma}{\alpha}-|\alpha-1|\log(\dim\hil).\label{old-new bounds10}
\end{align}
\end{Lemma}
\begin{proof}
According to the Araki-Lieb-Thirring inequality \cite{Araki,LT}, for any positive semidefinite operators 
$A,B$,
\begin{equation}\label{ALT}
s(\alpha)\Tr A^{\alpha}B^{\alpha}A^{\alpha}\le s(\alpha)\Tr (ABA)^{\alpha}.
\end{equation}
A converse to the Araki-Lieb-Thirring inequality was  given in 
\cite{Aud-ALT}, where it was shown that 
\begin{equation}\label{converse ALT1}
s(\alpha)\Tr (ABA)^{\alpha}\le s(\alpha)\bz\norm{B}^{\alpha}\Tr A^{2\alpha}\jz^{1-\alpha}
\bz\Tr A^{\alpha}B^{\alpha}A^{\alpha}\jz^{\alpha}.
\end{equation}
Applying \eqref{ALT} and \eqref{converse ALT1} to $A:=\rho^{\half}$ and 
$B:=\sigma^{\frac{1-\alpha}{\alpha}}$, we get 
\begin{align}
s(\alpha)\Tr\rho^{\alpha}\sigma^{1-\alpha}
&\le
s(\alpha)\Tr\bz\rho^{\half}\sigma^{\frac{1-\alpha}{\alpha}}\rho^{\half}\jz^{\alpha}\nonumber\\
&\le
s(\alpha)\norm{\sigma}^{(1-\alpha)^2}\bz\Tr\rho^{\alpha}\jz^{1-\alpha}\bz\Tr\rho^{\alpha}\sigma^{1-\alpha}\jz^{\alpha}.
\label{old-new bounds0}
\end{align}
This is equivalent to \eqref{old-new bounds} for invertible $\rho$ and $\sigma$, and 
hence \eqref{old-new bounds} holds also for general $\rho$ and $\sigma$ due to 
\eqref{def by limit}.

When $\alpha\in(0,1)$,
plugging \eqref{power bound} into the second inequality in \eqref{old-new bounds0} yields
\begin{align*}
\Tr\bz\rho^{\half}\sigma^{\frac{1-\alpha}{\alpha}}\rho^{\half}\jz^{\alpha}
\le&
\norm{\sigma}^{(1-\alpha)^2}\bz\Tr\rho^{0}\jz^{(1-\alpha)^2}\bz\Tr\rho\jz^{\alpha(1-\alpha)}\nonumber\\
&\cdot\bz\Tr\rho^{\alpha}\sigma^{1-\alpha}\jz^{\alpha},
\end{align*}
and hence
\begin{align*}
\rsrn{\rho}{\sigma}{\alpha}
\ge&
\alpha\rsro{\rho}{\sigma}{\alpha}\\
&+(1-\alpha)\bz\log\Tr\rho-\log\Tr\rho^{0}-\log\norm{\sigma}\jz.
\end{align*}
From this, \eqref{old-new bounds9} and \eqref{old-new bounds7} follow immediately.

When $\alpha>1$, we have $\Tr\bz\rho/\norm{\rho}\jz^{\alpha}\le\Tr\bz\rho/\norm{\rho}\jz$, and plugging it into \eqref{old-new bounds} yields
\begin{align*}
\rsrn{\rho}{\sigma}{\alpha}
\ge
\alpha\rsro{\rho}{\sigma}{\alpha}+(\alpha-1)\bz\log\norm{\sigma}-\log\norm{\rho}\jz,
\end{align*}
and \eqref{old-new bounds9} follows as a special case. 
In particular, if $\norm{\rho}\le 1$ then $\Tr\sigma\le\norm{\sigma}\Tr\sigma^0$ yields
\begin{align*}
\rsrn{\rho}{\sigma}{\alpha}
\ge
\alpha\rsro{\rho}{\sigma}{\alpha}+(\alpha-1)\bz\log\Tr\sigma-\log\Tr\sigma^0\jz,
\end{align*}
which yields \eqref{old-new bounds7}.
\end{proof}

\begin{cor}\label{cor:1 limit}
For any two non-zero positive semidefinite operators $\rho,\sigma$,
\begin{align}\label{1 limit new}
\lim_{\alpha\to 1}\rsrn{\rho}{\sigma}{\alpha}=
\rsr{\rho}{\sigma}{1}.
\end{align}
\end{cor}
\begin{proof}
Immediate from \eqref{old-new bounds} and \eqref{1 limit}.
\end{proof}

\begin{rem}
According to the results of \cite{Hiai-ALT}, the first inequality in \eqref{old-new bounds} holds as an equality if and only if
$\alpha=1$ or $\rho$ and $\sigma$ commute with each other.
\end{rem}

\begin{rem}\label{rem:quant}
A quantitative version of \eqref{1 limit} was given in \cite[Lemma 6.3]{Tomamichel} for $\alpha\searrow 1$, and the same argument yields analogous bounds for $\alpha\nearrow 1$, as noted in \cite[Lemma 2.3]{AMV}. A quantitative version of 
\eqref{1 limit new} can be obtained by combinig the bound in \cite[Lemma 2.3]{AMV} with the inequalities of 
Lemma \ref{Lemma:old-new}, which yields
\begin{align*}
\rsr{\rho}{\sigma}{1}\ge&\rsrn{\rho}{\sigma}{\alpha}\\
\ge&
\alpha\rsr{\rho}{\sigma}{1}-4\alpha(1-\alpha)(\log\eta)^2\cosh c\\
&+\log\Tr\rho
-\log\Tr\rho^{\alpha}+(1-\alpha)\log\norm{\sigma}\inv,
\end{align*}
when $1-\delta<\alpha<1$, and 
\begin{align*}
\rsr{\rho}{\sigma}{1}\le\rsrn{\rho}{\sigma}{\alpha}
\le&
\rsr{\rho}{\sigma}{1}-4(1-\alpha)(\log\eta)^2\cosh c,
\end{align*}
when $1<\alpha<1+\delta$,
where $\eta:=1+\Tr\rho^{3/2}\sigma^{-1/2}+\Tr\rho^{1/2}\sigma^{1/2}$, $c$ is an arbitrary positive number, and  
$\delta:=\min\left\{\half, \frac{c}{2\log\eta}\right\}$. 
The second set of inequalities has already been noted in \cite{WWY}. In particular, if $\rho$ and $\sigma$ are states
then using \eqref{old-new bounds7} instead of \eqref{old-new bounds} in the first set of inequalities above, we get
\begin{align*}
\rsr{\rho}{\sigma}{1}
\ge&
\rsrn{\rho}{\sigma}{\alpha}\\
\ge& 
\alpha\rsr{\rho}{\sigma}{1}\\
&-(1-\alpha)\left[ 4\alpha(\log\eta)^2\cosh c+\log(\dim\hil)\right]
\end{align*}
for every $1-\delta<\alpha<1$.
\end{rem}
\smallskip

We will also need the following generalization of \eqref{1 limit} and \eqref{1 limit new}:
\begin{Lemma}\label{Lemma:dist limit}
Let $\N\subseteq\S(\hil)$ and $\sigma\in\B(\hil)_+$ be such that $\supp\rho\subseteq\supp\sigma$ for all $\rho\in\N$. 
For both $\xx=\oldd$ and $\xx=\nww$,
\begin{align}\label{dist limit}
\lim_{\alpha\to 1}\inf_{\rho\in\N}D_{\alpha}\x(\rho\|\sigma)=\inf_{\rho\in\N}D_1(\rho\|\sigma).
\end{align}
\end{Lemma}
\begin{proof}
By \eqref{mon} and \eqref{1 limit}, we have 
\begin{align*}
\lim_{\alpha\searrow 1}\inf_{\rho\in\N}D_{\alpha}(\rho\|\sigma)&=
\inf_{\alpha>1}\inf_{\rho\in\N}D_{\alpha}(\rho\|\sigma)\\
&=
\inf_{\rho\in\N}\inf_{\alpha>1}D_{\alpha}(\rho\|\sigma)\\
&=
\inf_{\rho\in\N}D_1(\rho\|\sigma).
\end{align*}
Thanks to the support assumption, $\rho\mapsto D_{\alpha}(\rho\|\sigma)$ is continuous on $\N$ for every $\alpha\in(0,+\infty)$, and hence it is also continuous on the closure
(w.r.t.~any norm) $\ol\N$ of $\N$, and 
$\inf_{\rho\in\N}D_{\alpha}(\rho\|\sigma)=\min_{\rho\in\ol\N}D_{\alpha}(\rho\|\sigma)$.
Using again the monotonicity \eqref{mon}, Lemma \ref{Lemma:minimax} and \eqref{1 limit}, we have
\begin{align*}
\lim_{\alpha\nearrow 1}\inf_{\rho\in\N}D_{\alpha}(\rho\|\sigma)&=
\sup_{\alpha\in(0,1)}\min_{\rho\in\ol\N}D_{\alpha}(\rho\|\sigma)\\
&=
\min_{\rho\in\ol\N}\sup_{\alpha\in(0,1)}D_{\alpha}(\rho\|\sigma)\\
&=
\min_{\rho\in\ol\N}D_1(\rho\|\sigma)\\
&=
\inf_{\rho\in\N}D_1(\rho\|\sigma).
\end{align*}
This proves the assertion for $\xx=\oldd$. Using now \eqref{old-new bounds9}, we have 
\begin{align*}
\inf_{\rho\in\N}D_{\alpha}\old(\rho\|\sigma)
\ge& 
\inf_{\rho\in\N}D_{\alpha}\nw(\rho\|\sigma)\\
\ge&
\alpha\inf_{\rho\in\N}\rsro{\rho}{\sigma}{\alpha}
-|\alpha-1|\log\dim\hil\\
&+(\alpha-1)\log\norm{\sigma}.
\end{align*}
Combining it with \eqref{dist limit} for $\xx=\oldd$ yields \eqref{dist limit} for $\xx=\nww$.
\end{proof}

\subsection{Convexity properties}
\label{sec:convexity}

Probably the most important mathematical property of the R\'enyi divergences is their monotonicity under CPTP maps
for certain ranges of the parameter $\alpha$. This is known to be equivalent to the joint concavity of $s(\alpha)Q_{\alpha}\x$,
in the sense that they can be easily derived from each other. The latter can be formulated as follows:
If $\rho_i,\sigma_i\in\B(\hil)_+,\,i=1,\ldots,r$, and $\gamma_1,\ldots,\gamma_r$ is a probability distribution
on $[r]:=\{1,\ldots,r\}$, then
\begin{align}\label{joint concavity}
s(\alpha)Q\x_{\alpha}\bz\sum_i\gamma_i\rho_i\Big\|\sum_i\gamma_i\sigma_i\jz&\ge s(\alpha)\sum_i\gamma_iQ\x_{\alpha}(\rho_i\|\sigma_i)
\end{align}
for $\typ=\oldd$ and $\alpha\in[0,2]$ and for $\typ=\nww$ and $\alpha\in[1/2,+\infty)$
(for $\alpha>1$ one also has to assume that $\supp\rho_i\subseteq\supp\sigma_i$ for all $i$.)
This has been proved for $\typ=\oldd$ and $\alpha\in(0,1)$ in \cite{Lieb}, and 
for $\typ=\oldd$ and $\alpha\in(1,2]$ in \cite{Ando}; see also \cite{HMPB,Petz} for a different proof of both.
The case $\typ=\nww$ and $\alpha\in[1/2,1]$ follows from 
the general concavity result in \cite[Theorem 2.1]{Hiai}, and the case 
$\typ=\nww$ and $\alpha\in[1,2]$ was proved in \cite{Renyi_new,WWY}. Finally, the case 
$\typ=\nww$ was proved by a different method in \cite{FL} for all $\alpha\in[1/2,+\infty)$.
It is known that for $\typ=\oldd$ and $\alpha>2$, and for $\typ=\nww$ and $\alpha\in(0,1/2)$, \eqref{joint concavity}
need not hold in general \cite{Renyi_new}.

Our goal here is to complement \eqref{joint concavity} to some extent.
The following Lemma is a special case of the famous Rotfel'd inequality
(see, e.g., Section 4.5 in \cite{Hiaibook}). 
For the coding theorems in Sections \ref{sec:Stein}--\ref{sec:capacity}, we only need the inequality 
\eqref{subadditivity inequality1} below for $\alpha\in(0,1)$. For readers' convenience,
we include an elementary proof below that covers this range of $\alpha$.

\begin{Lemma}\label{Lemma:subadditivity1}
The function $A\mapsto s(\alpha)\Tr A^{\alpha}$ is subadditive on positive semidefinite operators for every 
$\alpha\in[0,+\infty)$. That is, if $A,B\in\B(\hil)_+$ then 
\begin{align}
s(\alpha)\Tr(A+B)^{\alpha}&\le s(\alpha)\bz\Tr A^{\alpha}+\Tr B^{\alpha}\jz,\ds\ds\alpha\in[0,+\infty).
\label{subadditivity inequality1}
\end{align}
\end{Lemma}
\begin{proof}
We only prove the case $\alpha\in[0,2]$.
Assume first that $A$ and $B$ are invertible and let $\alpha\in(0,1)$. Then
\begin{align*}
\Tr(A+B)^{\alpha}-\Tr A^{\alpha}&=
\int_{0}^1\frac{d}{dt}\Tr(A+tB)^{\alpha}\,dt\\
&=
\int_{0}^1\alpha\Tr B(A+tB)^{\alpha-1}\,dt\\
&\le
\int_{0}^1\alpha\Tr B(tB)^{\alpha-1}\,dt\\
&=\Tr B^{\alpha}\int_{0}^1\alpha t^{\alpha-1}\,dt\\
&=\Tr B^{\alpha},
\end{align*}
where in the first line we used
the identity $(d/dt)\Tr f(A+tB)=\Tr Bf'(A+tB)$, and the inequality follows
from the fact that $x\mapsto x^{\alpha-1}$ is operator monotone decreasing on $(0,+\infty)$ for 
$\alpha\in(0,1)$. 
This proves \eqref{subadditivity inequality1} for invertible $A$ and $B$, and
the general case follows by continuity. 
The proof for the case $\alpha\in(1,2]$ goes the same way, using the fact that 
$x\mapsto x^{\alpha-1}$ is operator monotone increasing on $(0,+\infty)$ for 
$\alpha\in(1,2]$. 
The case $\alpha=1$ is trivial, and the case $\alpha=0$ follows by taking the limit $\alpha\to 0$ in \eqref{subadditivity inequality1}.
\end{proof}

\begin{prop}\label{prop:complements}
Let $\sigma,\rho_1,\ldots,\rho_r\in\B(\hil)_+$, and $\gamma_1,\ldots,\gamma_r$ be a probability distribution on $[r]$. 
For every $\alpha\in[0,+\infty)$,
\begin{align}
s(\alpha)\sum_i\gamma_iQ\nw_{\alpha}(\rho_i\|\sigma)
&\le 
s(\alpha)Q\nw_{\alpha}\bz\sum_i\gamma_i\rho_i\Big\|\sigma\jz\nonumber\\
&\le 
s(\alpha)\sum_i\gamma_i^{\alpha}Q\nw_{\alpha}(\rho_i\|\sigma),\label{concavity complement}
\end{align}
and
\begin{align}
\max_i\rsrn{\rho_i}{\sigma}{\alpha}
&\ge
\bigrsrn{\sum_{i=1}^r\gamma_i\rho_i}{\sigma}{\alpha}\nonumber\\
&\ge
\min_i\rsrn{\rho_i}{\sigma}{\alpha}+\log \min_i\gamma_i.\label{concavity complement2}
\end{align} 
Moreover, the second inequalities in \eqref{concavity complement} and \eqref{concavity complement2} are valid for arbitrary non-negative $\gamma_1,\ldots,\gamma_r$ with 
$\gamma_1+\ldots+\gamma_r>0$.
\end{prop}
\begin{proof}
By Lemma \ref{Lemma:subadditivity1}, we have
\begin{align*}
\Tr\bz\sigma^{\frac{1-\alpha}{2\alpha}}\bz\sum_{i=1}^r\gamma_i\rho_i\jz\sigma^{\frac{1-\alpha}{2\alpha}}\jz^{\alpha}
&\le
\sum_{i=1}^r\Tr\bz\sigma^{\frac{1-\alpha}{2\alpha}}\gamma_i\rho_i\sigma^{\frac{1-\alpha}{2\alpha}}\jz^{\alpha}\\
&=
\sum_{i=1}^r\gamma_i^{\alpha}\Tr\bz\sigma^{\frac{1-\alpha}{2\alpha}}\rho_i\sigma^{\frac{1-\alpha}{2\alpha}}\jz^{\alpha}
\end{align*}
for $\alpha\in(0,1)$, and the inequality is reversed for $\alpha>1$, which proves the second inequality in \eqref{concavity complement}. The first inequality follows the same way, by noting that $A\mapsto\Tr A^{\alpha}$ is concave for $\alpha\in(0,1]$ and convex for $\alpha\ge 1$.

For the proof of \eqref{concavity complement2}, we may assume that $\rho$ and $\sigma$ are invertible, due to 
\eqref{def by limit}.
We prove the inequalities for $\alpha\in(0,1)$; the proof for $\alpha\in(1,+\infty)$ goes exactly the same way, and 
the cases $\alpha=0,1$ follow by taking the corresponding limit in $\alpha$.
We have
\begin{align*}
\bigrsrn{\sum_{i=1}^r\gamma_i\rho_i}{\sigma}{\alpha}&=
\frac{1}{\alpha-1}\log\frac{Q\nw_{\alpha}\bz\sum_i\gamma_i\rho_i\Big\|\sigma\jz}{\sum_i\gamma_i\Tr\rho_i}\\
&\le
\frac{1}{\alpha-1}\log\frac{\sum_i\gamma_iQ\nw_{\alpha}\bz\rho_i\|\sigma\jz}{\sum_i\gamma_i\Tr\rho_i}\\
&\le
\frac{1}{\alpha-1}\log\min_i\frac{Q\nw_{\alpha}\bz\rho_i\|\sigma\jz}{\Tr\rho_i}\\
&=
\max_i D_{\alpha}\nw(\rho_i\|\sigma),
\end{align*}
where the first inequality is due to the first inequality in \eqref{concavity complement}
(note that $\alpha-1<0$ by assumption),
and the second inequality follows from the trivial inequality 
$Q\nw_{\alpha}\bz\rho_j\|\sigma\jz\ge (\Tr\rho_j)\min_i\frac{Q\nw_{\alpha}\bz\rho_i\|\sigma\jz}{\Tr\rho_i}$
after multiplying both sides by $\gamma_j$ and summing over $j$.
This proves the first inequality in \eqref{concavity complement2}.

The second inequality in \eqref{concavity complement} yields
\begin{align*}
\bigrsrn{\sum_{i=1}^r\gamma_i\rho_i}{\sigma}{\alpha}&=
\frac{1}{\alpha-1}\log \frac{Q\nw_{\alpha}\bz\sum_i\gamma_i\rho_i\Big\|\sigma\jz}{\Tr\sum_i\gamma_i\rho_i}\\
&\ge
\frac{1}{\alpha-1}\log 
\frac{\sum_i\gamma_i^{\alpha} Q\nw_{\alpha}\bz\rho_i\|\sigma\jz}{\sum_i\gamma_i\Tr\rho_i}.
\end{align*}
We have
\begin{align*}
\gamma_i^{\alpha} Q\nw_{\alpha}\bz\rho_i\|\sigma\jz
&\le
(\gamma_i^{\alpha}\Tr\rho_i)
\max_j\frac{\gamma_j^{\alpha} Q\nw_{\alpha}\bz\rho_j\|\sigma\jz}{\gamma_j^{\alpha}\Tr\rho_j}\\
&\le
\gamma_i\Tr\rho_i\bz\max_j\gamma_j^{\alpha-1}\jz
\max_j\frac{Q\nw_{\alpha}\bz\rho_j\|\sigma\jz}{\Tr\rho_j},
\end{align*}
and summing over $i$ and using again that $\alpha-1<0$, we obtain  
\begin{align*}
\frac{1}{\alpha-1}\log 
\frac{\sum_i\gamma_i^{\alpha} Q\nw_{\alpha}\bz\rho_i\|\sigma\jz}{\Tr\sum_i\gamma_i\rho_i}
\ge&
\min_j\frac{1}{\alpha-1}\log \frac{Q\nw_{\alpha}\bz\rho_j\|\sigma\jz}{\Tr\rho_j}\\
&+\log\min_j\gamma_j,
\end{align*}
which is exactly the second inequality in \eqref{concavity complement2}.
\end{proof}

\begin{rem}
Note that \eqref{joint concavity} expresses joint concavity,
whereas in Proposition \ref{prop:complements} we only took a convex 
combination in the first variable and not in the second. It is easy to see that this restriction is in fact necessary. Indeed, let $\rho_1:=\sigma_2:=\pr{x}$ and $\rho_2:=\sigma_1:=\pr{y}$, where $x$ and $y$ are orthogonal unit vectors in some Hilbert space. If we choose $\gamma_1=\gamma_2=1/2$ then $\sum_{i}\gamma_i\rho_i=\sum_i\gamma_i\sigma_i$, and hence
\begin{align*}
&\bigrsrn{\sum_{i=1}^r\gamma_i\rho_i}{\sum_{i=1}^r\gamma_i\sigma_i}{\alpha}=0,\ds\ds\text{while}\\
&\rsrn{\rho_1}{\sigma_1}{\alpha}=\rsrn{\rho_2}{\sigma_2}{\alpha}=+\infty.
\end{align*}
Thus, no inequality of the form $\bigrsrn{\sum_{i=1}^r\gamma_i\rho_i}{\sum_{i=1}^r\gamma_i\sigma_i}{\alpha}\ge 
c_1\min_i\rsrn{\rho_i}{\sigma_i}{\alpha}-c_2$ can hold for any positive constants $c_1$ and $c_2$.

Note also that the first inequality in \eqref{concavity complement} is a special case of the joint concavity inequality
\eqref{joint concavity} for $\alpha\ge 1/2$, but not for the range 
$0<\alpha<1/2$, where joint concavity fails \cite{Renyi_new}. Here again it is important that we took a convex combination only in the first variable of $Q_{\alpha}\nw$.
\end{rem}

\begin{rem}
The same example as in \cite{Tikhonov1,Tikhonov2} shows that 
the power functions $x\mapsto s(\alpha)x^{\alpha}$ are not operator subadditive for any $\alpha\ne 1$, 
i.e., \eqref{subadditivity inequality1} cannot hold without taking the trace.
In fact, for any given $\alpha\in(0,+\infty)\setminus\{1\}$ and any negative number $\nu$, there exist $A,B\in\B(\bC^2)$ such that 
$s(\alpha)(A^{\alpha}+B^{\alpha}-(A+B)^{\alpha})$ has an eigenvalue below $\nu$.
As a consequence, $s(\alpha)Q_{\alpha}\old$ doesn't satisfy a subadditivity inequality similar to the one in 
\eqref{concavity complement} for any $\alpha\ne 1$. However, combining \eqref{concavity complement} with Lemma \ref{Lemma:old-new}, we get
\begin{align*}
&s(\alpha)Q\old_{\alpha}\bz\sum_i\gamma_i\rho_i\Big\|\sigma\jz\\
&\ds\le s(\alpha)\sum_i\gamma_i^{\alpha}Q\old_{\alpha}(\rho_i\|\sigma)^{\alpha}
\norm{\sigma}^{(1-\alpha)^2}(\Tr\rho_i^{\alpha})^{1-\alpha},
\end{align*}
from which it is easy to obtain the inequality
\begin{align*}
D\old_{\alpha}\bz\sum_i\gamma_i\rho_i\Big\|\sigma\jz
\ge &
\alpha\min_i D\old_{\alpha}(\rho_i\|\sigma)
+(\alpha-1)\log\norm{\sigma}\\
&+\log\min_i\left\{\gamma_i\frac{\Tr\rho_i}{\Tr\rho_i^{\alpha}}\right\}
\end{align*}
for all $\alpha\in[0,+\infty)$. When all the $\rho_i$ and $\sigma$ are states on $\hil$, then
combining \eqref{concavity complement2} with \eqref{old-new bounds7} yields
\begin{align*}
D\old_{\alpha}\bz\sum_i\gamma_i\rho_i\Big\|\sigma\jz
\ge &
\alpha\min_i D\old_{\alpha}(\rho_i\|\sigma)\\
&+\log\min_i\gamma_i-|\alpha-1|\log\dim\hil.
\end{align*}
Note that this is a non-trivial inequality even for $\alpha=1$.
\end{rem}

\subsection{R\'enyi capacities}

By a \ki{classical-quantum channel}, or simply a \ki{channel}, $W$ we mean a map $W:\,\X\to\S(\hil)$, where $\X$ is some input alphabet (which can be an arbitrary non-empty set) and $\hil$ is a finite-dimensional Hilbert space. 
We recover the usual notion of a \ki{quantum channel} when $\X=\S(\kil)$ for some Hilbert space $\kil$, and $W$ is a completely positive trace-preserving linear map. A channel $W$ is called \ki{classical} if all the $W(x)$ commute with each other
for every $x\in\X$.

For an input alphabet $\X$, let $\{\delta_x\}_{x\in\X}$ be a set of rank-$1$ orthogonal projections in some Hilbert space $\hil_{\X}$, and
for every channel $W:\,\X\to\S(\hil)$ define
\begin{equation*}
\what W:\,x\mapsto \delta_x\otimes W(x).
\end{equation*}

\begin{rem}
Note that if $\X$ is of infinite cardinality then $\hil_{\X}$ and $\hil_{\X}\otimes\hil$ are infinite-dimensional.  
The state space (the set of density operators) $\S(\kil)$ of an infinite-dimensional Hilbert space $\kil$ is defined to be the 
set of positive semidefinite trace-class operators on $\kil$ with trace $1$. We further introduce the notation 
$\S_f(\kil)$ for the set of finite-rank density operators on $\kil$.
Since $\hil$ is finite-dimensional, we have $\what W(x)\in\S_f(\hil_{\X}\otimes\hil)$ for every $x\in\X$.

In the following, we will consider R\'enyi divergences of the form 
$D_{\alpha}\x(\rho\|\sigma)$ for $\rho,\sigma\in\S_f(\hil_{\X}\otimes\hil)$. Since the operators are of finite rank, one can always restrict the Hilbert space to their joint support and assume that the Hilbert space is finite-dimensional. Hence, the R\'enyi divergences are well-defined, and the results of the previous sections can be used without alteration.
\end{rem}

Let $\P_f(\X)$ denote the set of finitely supported probability measures on $\X$.
The maps $W$ and $\what W$ can naturally be extended to convex maps $W:\,\P_f(\X)\to\S(\hil)$ and $\what W:\,\P_f(\X)\to\S_f(\hil_{\X}\otimes\hil)$, as
\begin{align*}
W(p)&:=\sum_{x\in\X}p(x)W(x),\\
\what W(p)&:=\sum_{x\in\X}p(x)\what W(p)=\sum_{x\in\X}p(x)\delta_x\otimes W(x).
\end{align*}
Note that $\what W(p)$ is a classical-quantum state, and the marginals of $\what W(p)$ are given by
\begin{align*}
\Tr_{\hil}\what W(p)&=\hat p:=\sum_x p(x)\delta_x\ds\ds\ds\ds\ds\text{and}\\
\Tr_{\hil_{\X}}\what W(p)&=W(p).
\end{align*}

For a channel $W:\,\X\to\S(\hil)$, and a probability distribution $p\in\P_f(\X)$, the corresponding \ki{Holevo quantity} $\chi(W,p)$ is the 
\ki{mutual information} in the classical-quantum state $\what W(p)$, defined as
\begin{align}
\chi(W,p)&:=\chi_1(W,p)\nonumber\\
&:=\rsr{\what W(p)}{\hat p\otimes W(p)}{1}
\label{Hol def}\\
&=\inf_{\rho\in\S(\hil_{\X}),\sigma\in\S(\hil)}D_1\bz \what W(p)\|\rho\otimes\sigma\jz
\label{Hol def2}\\
&=
\inf_{\rho\in\S(\hil_{\X})}\rsr{\what W(p)}{\rho\otimes W(p)}{1}\label{Hol def3}\\
&=\inf_{\sigma\in\S(\hil)}\rsr{\what W(p)}{\hat p\otimes \sigma}{1},\label{Hol def4}
\end{align}
where $D_1$ is the relative entropy \eqref{1 limit}, and the equality of the expressions in 
\eqref{Hol def}--\eqref{Hol def4} is easy to verify from the non-negativity of the relative entropy on pairs of states.
The \ki{Holevo capacity} $\chi(W)$ is the maximal mutual information over all possible input distributions, i.e.,
\begin{align}\label{Holevo cap}
\chi(W)&:=\sup_{p\in\P_f(\X)}\chi(W,p).
\end{align}

By the Holevo-Schumacher-Westmoreland theorem \cite{Holevo,SW}, $\chi(W)$ is the optimal rate at which classical information can be sent through the channel with asymptotically vanishing error; see Section \ref{sec:capacity} for details. It is also known that the asymptotic behaviour of the decoding error probability for rates below or above the Holevo capacity can be described by
the $\alpha$-capacities of the channel; see \cite{Csiszar} for the case of classical channels, and \cite{MO3} for the case of classical-quantum channels in the strong converse domain. Below we give the definition of the $\alpha$-capacities, and collect a few properties that we will need in Section \ref{sec:capacity}.
\smallskip

If we replace $D_1$ with some $D_{\alpha}\x$ with $\alpha\ne 1$ then the expressions in 
\eqref{Hol def}--\eqref{Hol def4} need not be equal anymore, and we choose the one in 
\eqref{Hol def4} to define the \ki{$\alpha$-mutual information} in $\what W(p)$ as 
\begin{align}\label{alpha mutual}
\hol{\alpha}\x(W,p):=\inf_{\sigma\in\S(\hil)} D_{\alpha}\x\bz\what W(p)\|\hat p\otimes \sigma\jz,
\end{align}
where $\typ=\oldd$ or $\typ=\nww$, and $\alpha\in(0,+\infty)$. The corresponding \ki{$\alpha$-capacities} are then defined as
\begin{align}\label{alpha cap}
\hol{\alpha}\x(W):=\sup_{p\in\P_f(\X)}\hol{\alpha}\x(W,p).
\end{align}

\begin{rem}
Choosing to optimize only over the state of the output system in \eqref{alpha mutual} might seem somewhat arbitrary, especially when 
compared to the more symmetric forms in \eqref{Hol def} and \eqref{Hol def2}. There are various reasons, though, to prefer this 
seemingly less natural optimization. One is the additivity properties \eqref{chi additivity} and \eqref{chi additivity2}, which are crucial for applications, and which are not known (at least to the author) to hold with the types of optimization in \eqref{Hol def2} and \eqref{Hol def3}.
Another is that
the capacity formula \eqref{alpha cap}, based on \eqref{alpha mutual} has an operational interpretation (for $\alpha\ge 1/2$)
as a generalized cutoff-rate \cite{Csiszar}, showing that this is probably the right (in the sense of operationally justified)
notion of $\alpha$-capacity, at least for classical channels, where $\hol{\alpha}\nw(W)=\hol{\alpha}\old(W)$.
A recent result \cite{MO3} shows that the same operational interpretation holds for $\hol{\alpha}\nw(W)$ and $\alpha\ge 1$ in the case of classical-quantum channels. No such operational interpretations are known for the $\alpha$-capacities based on the optimizations in \eqref{Hol def}--\eqref{Hol def3}.
\end{rem}

As it was pointed out in \cite{KW,Sibson}, and is easy to verify,
\begin{align}
\rsro{\what W(p)}{\hat p\otimes \sigma}{\alpha}=&
\frac{\alpha}{\alpha-1}\log\Tr\omega(W,p)\nonumber\\
&+D_{\alpha}\old\bz \bar\omega(W,p)\|\sigma\jz\label{Sibson}
\end{align}
for any state $\sigma$, 
where 
$\bar\omega(W,p):=\omega(W,p)/\Tr\omega(W,p)$ and
$\omega(W,p):=\bz\sum_x p(x)W(x)^{\alpha}\jz^{\frac{1}{\alpha}}$.
Since $D_{\alpha}\old$ is non-negative on pairs of density operators, we get
\begin{align}
\hol{\alpha}\old(W,p)&=\frac{\alpha}{\alpha-1}\log\Tr\omega(W,p)\nonumber\\
&=
\frac{\alpha}{\alpha-1}\log\Tr\bz\sum_x p(x)W(x)^{\alpha}\jz^{\frac{1}{\alpha}}.\label{explicit alpha-Holevo}
\end{align}
No such explicit formula is known for $\hol{\alpha}\nw(W,p)$.
\smallskip

Monotonicity of $D_{\alpha}$ in $\alpha$ yields that 
$\hol{\alpha}\old(W,p)$ 
is also monotonic increasing in $\alpha$. 
A simple minimax argument shows (see, e.g.~\cite[Lemma B.3]{MH}) that 
\begin{equation}\label{chi limit}
\lim_{\alpha\to 1}\hol{\alpha}\old(W,p)
=\chi(W,p),
\end{equation}
where $\chi(W,p)$ is the Holevo quantity. 
We will need the following generalization of this in Section \ref{sec:capacity}:

\begin{Lemma}\label{Lemma:inf Hol limit}
Let $W_i:\,\X\to\S(\hil),\,i\in\I$, be a set of channels, with some arbitrary index set $\I$, 
and let $p\in\P_f(\X)$ be a finitely supported probability distribution on $\X$. Then
\begin{align*}
\lim_{\alpha\to 1}\inf_{i\in\I}\hol{\alpha}(W_i,p)=\inf_{i\in\I}\chi(W_i,p).
\end{align*}
\end{Lemma}
\begin{proof}
It is easy to see from the explicit formulas \eqref{Hol def} and \eqref{explicit alpha-Holevo} that
the values of $\hol{\alpha}(W_i,p)$ only depend on the values of $W_i$ at the points of $\supp p$, which is, by assumption, a 
finite set. Hence, we can assume without loss of generality that $\X$ is finite, and therefore the vector space
 of functions from $\X$ to $\B(\hil)$, denoted by $\B(\hil)^{\X}$,
 is finite-dimensional.
Taking any norm on $\B(\hil)^{\X}$, the closure $C$ of $\{W_i\}_{i\in\I}$ is compact, and 
\eqref{Hol def} and \eqref{explicit alpha-Holevo} show that $W\mapsto \hol{\alpha}(W,p)$ is continuous on $C$ for every 
$\alpha\in(0,+\infty)$. Since 
$\alpha\mapsto\hol{\alpha}(W_i,p)$ is monotone increasing in $\alpha$, the same argument as in the proof of Lemma \ref{Lemma:dist limit} yields the assertion.
\end{proof}
\smallskip

We close this section with a few observations about the $\alpha$-capacities. Although we will not need these 
for the coding theorems presented later, they might be interesting for future applications.

First, note that $\max\{\Tr\what W(p)^{0},\Tr(\hat p\otimes \sigma)^0\}\le|\supp p|\dim\hil$, 
where $|\supp p|$ denotes the cardinality of the support of $p$,
and \eqref{old-new bounds7} yields that
\begin{align}
\hol{\alpha}\old(W,p)
&\ge
\hol{\alpha}\nw(W,p)\nonumber\\
&\ge
\alpha\hol{\alpha}\nw(W,p)-|\alpha-1|\log\bz|\supp p|\dim\hil\jz\label{weak capacity bounds}
\end{align}
for every $\alpha\in(0,+\infty)$.
Hence, in the setting of Lemma \ref{Lemma:inf Hol limit}, we also have
\begin{align*}
\lim_{\alpha\to 1}\inf_{i\in\I}\hol{\alpha}\nw(W_i,p)=\inf_{i\in\I}\chi(W_i,p).
\end{align*}

Next, we consider the limit of the $\alpha$-capacities as $\alpha\to 1$.
It was shown in \cite[Proposition B.5]{MH} that if $\ran W:=\{W(x)\,:\,x\in\X\}$ is compact then 
\begin{equation}\label{alpha cap lim}
\lim_{\alpha\to 1}\hol{\alpha}\old(W)=\chi(W).
\end{equation}
To obtain the same limit relation for $\hol{\alpha}\nw(W)$, we will need the following improvement of
\eqref{weak capacity bounds}:

\begin{Lemma}\label{Lemma:chi bounds}
Let $W:\,\X\to\S(\hil)$ be a channel, and $\alpha\in(0,+\infty)$.
For any $p\in\P_f(\X)$ and any $\sigma\in\S(\hil)$, we have
\begin{align}
\rsrn{\what W(p)}{\hat p\otimes \sigma}{\alpha}
\ge&
\alpha\rsro{\what W(p)}{\hat p\otimes \sigma}{\alpha}\nonumber\\
&-|\alpha-1|\log(\dim\hil),\label{Lemma:chi bounds2}
\end{align}
and hence,
\begin{align}\label{Lemma:chi bounds3}
\hol{\alpha}\old(W,p)\ge
\hol{\alpha}\nw(W,p)\ge
\alpha\hol{\alpha}\old(W,p)-|\alpha-1|\log(\dim\hil).
\end{align}
\end{Lemma}
\begin{proof}
First note that we can assume without loss of generality that 
$\supp \what W(p)\subseteq\supp(\hat p\otimes \sigma)$, since otherwise 
\eqref{Lemma:chi bounds2} holds trivially.
Let us fix $\alpha>1$. 
By \eqref{old-new bounds10} we have, for every $x\in\X$, that
$\Tr\bz W(x)^{\half}\sigma^{\frac{1-\alpha}{\alpha}}W(x)^{\half}\jz^{\alpha}\ge (\dim\hil)^{-(\alpha-1)^2}\bz\Tr W(x)^{\alpha}\sigma^{1-\alpha}\jz^{\alpha}$,  and hence,
\begin{align*}
&\rsrn{\what W(p)}{\hat p\otimes \sigma}{\alpha}\\
&=
\frac{1}{\alpha-1}\log\sum_x p(x)\Tr\bz W(x)^{\half}\sigma^{\frac{1-\alpha}{\alpha}}W(x)^{\half}\jz^{\alpha}\\
&\ge
\frac{1}{\alpha-1}\log\sum_x p(x)\bz\Tr W(x)^{\alpha}\sigma^{1-\alpha}\jz^{\alpha}\\
&\ds-(\alpha-1)\log(\dim\hil)\\
&\ge
\frac{1}{\alpha-1}\log\bz\sum_x p(x)\Tr W(x)^{\alpha}\sigma^{1-\alpha}\jz^{\alpha}\\
&\ds\s-(\alpha-1)\log(\dim\hil)\\
&=
\alpha\rsro{\what W(p)}{\hat p\otimes \sigma}{\alpha}-(\alpha-1)\log(\dim\hil),
\end{align*}
where the second inequality is due to the convexity of $x\mapsto x^{\alpha}$. The proof for $\alpha\in(0,1)$ goes exactly the same way. This proves \eqref{Lemma:chi bounds2}, and taking the infimum in $\sigma$ yields \eqref{Lemma:chi bounds3}.
\end{proof}

Lemma \ref{Lemma:chi bounds} and \eqref{alpha cap lim} yield immediately that
\begin{equation}\label{new capacity limit}
\lim_{\alpha\to 1}\hol{\alpha}\nw(W)=\chi(W).
\end{equation}

\begin{rem}
Carath\'eodory's theorem and the explicit formula \eqref{explicit alpha-Holevo} imply that in the definition 
$\hol{\alpha}\old(W):=\sup_{p\in\P_f(\X)}\hol{\alpha}\old(W,p)$
it is enough to consider probability distributions with $|\supp p|\le(\dim\hil)^2+1$. However, this is not known for $\hol{\alpha}\nw(W)$, and hence 
\eqref{weak capacity bounds} is insufficient to derive \eqref{new capacity limit}.
\end{rem}

\begin{rem}
For quantum channels, the limit relation $\lim_{\alpha\searrow 1}\hol{\alpha}\nw(W)=\chi(W)$ was proved
by a very different method in \cite{WWY}.
\end{rem}
\medskip

Finally, we point out a connection between $\alpha$-capacities and a special case of a famous convexity result by Carlen and Lieb \cite{CL,CL2}. 
For any finite-dimensional Hilbert space $\hil$ and 
$A_1,\ldots,A_n\in\B(\hil)_+$, define
\begin{equation*}
\Phi_{\alpha,q}(A_1,\ldots,A_n):=\bz\Tr\left[\bz\sum_{i=1}^n A_i^{\alpha}\jz^{q/\alpha}\right]\jz^{1/q},
\end{equation*}
$\alpha\ge 0,\,q> 0$.
Theorem 1.1 in \cite{CL2} says that for any finite-dimensional Hilbert space $\hil$, $\Phi_{\alpha,q}$ is concave
on $\bz\B(\hil)_+\jz^n$ for $0\le \alpha\le q\le 1$, and convex for all $1\le \alpha\le 2$ and $q\ge 1$.
Below we give an elementary proof of the following weaker statement:
$\Phi_{\alpha,1}^{\alpha}$ is concave for $\alpha\in(0,1)$ and convex for $\alpha\in(1,2]$.

For a set $\X$, a
finitely supported non-negative function $p:\,\X\to\bR_+$,
and a finite-dimensional Hilbert space $\hil$,
let $\hat\Phi_{p,\hil,\alpha}:\,\bz\B(\hil)_+\jz^{\X}\to\bR_+$ be defined as
\begin{equation*}
\hat\Phi_{p,\hil,\alpha}(W):=\bz\Tr\bz\sum_{x\in\X}p(x)W(x)^{\alpha}\jz^{1/\alpha}\jz^\alpha,
\end{equation*}
for every $W\in\bz\B(\hil)_+\jz^{\X}$.
The following Proposition is equivalent to our assertion:

\begin{prop}\label{prop:CL}
For any $\X,\,p$ and $\hil$, $\hat\Phi_{p,\hil,\alpha}$
is concave on $\bz\B(\hil)_+\jz^{\X}$ 
for $\alpha\in(0,1)$ and convex for $\alpha\in(1,2]$.
\end{prop}
\begin{proof}
Exactly the same way as in \eqref{Sibson}--\eqref{explicit alpha-Holevo}, we can see that 
\begin{align}
&\frac{\alpha}{\alpha-1}\log\Tr\bz\sum_x p(x)W(x)^{\alpha}\jz^{\frac{1}{\alpha}}\nonumber\\
&\ds=
\min_{\sigma\in\S(\hil)}D_{\alpha}\old\bz\what W(p)\|\hat p\otimes \sigma\jz.\label{khi representation}
\end{align} 
Assume for the rest that $\alpha\in(1,2]$; the proof for the case $\alpha\in(0,1)$ goes exactly the same way. 
Let $r\in\bN$, $W_1,\ldots,W_r\in(\B(\hil)_+)^{\X}$, and 
$\gamma_1,\ldots,\gamma_r$ be a probability distribution. Then
\begin{align*}
&\hat\Phi_{p,\hil,\alpha}\bz\sum_i\gamma_i W_i\jz\\
&\ds=
\min_{\sigma\in\S(\hil)}Q_{\alpha}\old\bz\sum_i\gamma_i\what W(p)\Big\|\hat p\otimes \sigma\jz\\
&\ds=
\min_{\sigma_1,\ldots,\sigma_r\in\S(\hil)}Q_{\alpha}\old\bz\sum_i\gamma_i\what W(p)\Big\|\hat p\otimes \sum_i\gamma_i\sigma_i\jz\\
&\ds\le
\min_{\sigma_1,\ldots,\sigma_r\in\S(\hil)}\sum_i\gamma_i Q_{\alpha}\old\bz\what W(p)\|\hat p\otimes\sigma_i\jz\\
&\ds=
\sum_i\gamma_i\min_{\sigma_i}Q_{\alpha}\old\bz\what W(p)\|\hat p\otimes\sigma_i\jz\\
&\ds=
\sum_i\gamma_i\hat\Phi_{p,\hil,\alpha}\bz W_i\jz,
\end{align*} 
where the first and the last identities are due to \eqref{khi representation}, and
the inequality follows from the joint convexity of $Q_{\alpha}\old$
 \cite{Ando,Petz}. (In the case $\alpha\in(0,1)$, we have to use joint concavity \cite{Lieb,Petz}.)
\end{proof}

\section{Coding theorems}
\label{sec:applications}

\subsection{Quantum Stein's Lemma with composite null-hypothesis}
\label{sec:Stein}

Consider the asymptotic hypothesis testing problem with null-hypothesis $H_0:\,\N_n\subset\S(\hil_n)$ and alternative hypothesis
$H_1:\,\sigma_n\in\S(\hil_n)$, $n\in\bN$, where $\hil_n$ is some finite-dimensional Hilbert space. 
Our goal is to decide between these two hypotheses based on the outcome of a binary 
POVM $(T_n(0),T_n(1))$ on $\hil_n$, where $0$ and $1$ indicate
the acceptance of $H_0$ and $H_1$, respectively.
Since $T_n(1)=I-T_n(0)$, the POVM is uniquely determined by $T_n=T_n(0)$, and 
the only constraint on $T_n$ is that $0\le T_n\le I_n$. We will call such operators \ki{tests}.
Given a test $T_n$, the probability of mistaking $H_0$ for $H_1$ (type I error) and
the probability of mistaking $H_1$ for $H_0$ (type II error) are given by
\begin{align*}
\alpha_n(T_n)&:=\sup_{\rho_n\in\N_n}\Tr\rho_n(I-T_n),\ds\text{(type I)},\ds\ds\text{and}\\
\beta_n(T_n)&:=\Tr\sigma_n T_n,\ds\text{(type II)}.
\end{align*}
\begin{defin}\label{def:direct rate}
We say that a rate $R\ge 0$ is \ki{achievable} if 
there exists a sequence of tests $T_n,\,n\in\bN$, with
\begin{align*}
\lim_{n\to+\infty}\alpha_n(T_n)=0\ds\ds\ds\text{and}\ds\ds\ds
\limsup_{n\to+\infty}\frac{1}{n}\log\beta_n(T_n)\le -R.
\end{align*} 
The largest achievable rate $R(\{\N_n\}_{n\in\bN}\|\{\sigma_n\}_{n\in\bN})$ is the \ki{direct rate} of the hypothesis testing problem.
\end{defin}

For the bigger part of this section, we assume that $\hil_n=\hil^{\otimes n},\,n\in\bN$, where $\hil=\hil_1$, and that the alternative hpothesis is i.i.d., i.e., $\sigma_n=\sigma^{\otimes n},\,n\in\bN$, 
with $\sigma=\sigma_1$. We say that the null-hypothesis is \ki{composite i.i.d.} if there exists a set $\N\subset\S(\hil)$ such that for all $n\in\bN$,
$\N_n=\N\notimes:=\{\rho^{\otimes n}:\,\rho\in\N\}$. The null-hypothesis is  \ki{simple i.i.d.} if $\N$ consists of one single element, i.e.,
$\N=\{\rho\}$ for some $\rho\in\S(\hil)$. According to the quantum Stein's Lemma \cite{HP,ON2}, the direct rate in the simple i.i.d.~case is given by $D_1(\rho\|\sigma)$.
The case of the general composite null-hypothesis was treated in \cite{BDKSSSz} under the name of quantum Sanov theorem. There it was shown that there exists a sequence of tests
$\{T_n\}_{n\in\bN}$ such that $\lim_{n\to+\infty}\Tr\rho^{\otimes n}(I-T_n)=0$ for every $\rho\in\N$, and 
$\limsup_{n\to+\infty}\frac{1}{n}\log\beta_n(T_n)\le -D_1(\N\|\rho)$, where 
$D_1(\N\|\rho):=\inf_{\rho\in\N}D_1(\rho\|\sigma)$. Note that this is somewhat weaker than 
$D_1(\N\|\rho)$ being achievable in the sense of Definition \ref{def:direct rate}.
Achievability in this stronger sense has been shown very recently in \cite{Notzel}, using the representation theory of the symmetric group and the method of types.
The proof in both papers followed the approach 
in \cite{HP} of reducing the problem to a classical hypothesis testing problem by projecting all states onto the commutative algebra generated by
$\{\sigma^{\otimes n}\}_{n\in\bN}$. 

Below we use a different proof technique to show that $D_1(\N\|\rho)$ is achievable in the sense of Defintion 
\ref{def:direct rate}. Our proof is based solely on Audenaert's trace inequality (Lemma \ref{Lemma:Aud}) and the subadditivity property of $Q_{\alpha}\nw$, given in 
Proposition \ref{prop:complements}. We obtain explicit upper bounds on the error probabilities for any finite $n\in\bN$ for a sequence of Neyman-Pearson type tests.
Moreover, if a $\delta$-net can be explicitly constructed for $\N$ for every $\delta>0$ (this is trivially satisfied when $\N$ is finite) then the tests can also be constructed 
explicitly. In \cite{BDKSSSz}, Stein's Lemma was stated with weak converse, while the results of \cite{Notzel} imply a strong converse. Here we use Nagaoka's method to further 
strengthen the converse part by giving exlicit bounds on the exponential rate with which the worst-case type I success probability goes to zero when the 
type II error decays with a rate larger than the optimal rate $D_1(\N\|\rho)$. 

Note that our proof technique doesn't actually rely on the i.i.d.~assumption, as we demonstrate in Theorem \ref{thm:correlated Stein}, where we give 
achievability bounds in the general correlated scenario. However, in the most general case we have to restrict to a finite null-hypothesis.
We show examples in Remark \ref{rem:correlated Stein} where the achievable rate of Theorem \ref{thm:correlated Stein} can be expressed as the 
regularized relative entropy distance of the null-hypothesis and the alternative hypothesis, giving a direct generalization of the i.i.d.~case. 
These results complement those of \cite{correlated-Sanov}, where it was shown that 
if $\Theta$ is a set of ergodic states on a spin chain, and $\Phi$ is a state 
on the spin chain such that for every $\Psi\in\Theta$,
Stein's Lemma holds for the simple hypothesis testing problem $H_0:\,\Psi,\,H_1:\,\Phi$, 
then it also holds for the composite hypothesis testing problem $H_0:\,\Theta,\,H_1:\,\Phi$. This was also extended
in \cite{correlated-Sanov}
to the case where $\Theta$ consists of translation-invariant states, using ergodic decomposition.
\medskip

Now let $\N\subset\S(\hil)$ be a non-empty set of states, and let $\sigma\in\B(\hil)_+$ be a positive semidefinite 
operator such that 
\begin{equation}\label{support condition}
\supp\rho\subseteq\supp\sigma,\ds\ds\ds \rho\in\N.
\end{equation}
Note that 
in hypothesis testing, $\sigma$ is usually assumed to be a state on $\hil$; however,
the proof for Stein's Lemma works the same way for a general positive semidefinite $\sigma$, and
considering this more general case is actually useful e.g., for state compression. 
Let
\begin{equation}\label{psi}
\psi\nw(t):=\sup_{\rho\in\N}\log Q_t\nw(\rho\|\sigma),\ds\ds\ds t>0,
\end{equation}
and for every $a\in\bR$, let
\begin{align}
\vfi\nw(a)&:=\sup_{0<t\le 1}\{at-\psi\nw(t)\},\nonumber\\
\hat\vfi\nw(a)&:=\sup_{0<t\le 1}\{a(t-1)-\psi\nw(t)\}=\vfi\nw(a)-a.\label{phi}
\end{align}
Note that $\vfi\nw$ is the Legendre-Fenchel transform of $\psi\nw$ on $(0,1]$.

\begin{thm}\label{thm:Stein}
For every $n\in\bN$, let $\N(n)\subset\N$ be a finite subset, and let 
$\delta(N(n)):=\sup_{\rho\in\N}\inf_{\rho'\in\N(n)}\norm{\rho-\rho'}_1$.
For every $a\in\bR$, let
$S_{n,a}:=\left\{e^{-na}\sum_{\rho\in\N(n)}\rho^{\otimes n}-\sigma^{\otimes n}>0\right\}$ be a Neyman-Pearson test.
Then
\begin{align}
\sup_{\rho\in\N}\Tr\rho^{\otimes n}(I-S_{n,a})&\le |\N(n)|e^{-n\hat\vfi\nw(a)}+n\delta(N(n)),
\label{finiten bounds1}\\
\Tr\sigma^{\otimes n}S_{n,a}&\le|\N(n)|e^{-n\vfi\nw(a)}.\label{finiten bounds}
\end{align}
In particular, let $\delta_n:=e^{-n\kappa}$ for some $\kappa>0$, and $\N(n):=\N_{\delta_n}\subset\N$ as in Lemma \ref{Lemma:state approximation}, with $V:=\B(\hil)\sa$ equipped with the trace-norm, and let $\Delta:=\dim_{\bR}V$. Then
\begin{align}
\limsup_{n\to+\infty}\frac{1}{n}\log\alpha_n(S_{n,a})&\le -\min\{\kappa,\hat\vfi\nw(a)-\kappa \Delta\},\label{NP upper1}\\
\limsup_{n\to+\infty}\frac{1}{n}\log\beta_n(S_{n,a})&\le -(\vfi\nw(a)-\kappa \Delta).\label{NP upper2}
\end{align}
\end{thm}
\begin{proof}
For every $n\in\bN$, let $\bar\rho_n:=\sum_{\rho\in\N(n)}\rho^{\otimes n}$, $\sigma_n:=\sigma^{\otimes n}$. 
Applying Lemma \ref{Lemma:Aud} to $A:=e^{-na}\bar\rho_n$ and $B:=\sigma_n$ for some fixed $a\in\bR$, we get 
\begin{align}
e_n(a)&:=e^{-na}\Tr\bar\rho_n (I-S_{n,a})+\Tr\sigma_nS_{n,a}\nonumber\\
&\le
e^{-nat}\Tr\bar\rho_n^t\sigma_n^{1-t}\label{ena upper bound}
\end{align}
for every $t\in[0,1]$. This we can further upper bound as
\begin{align}
\Tr\bar\rho_n^t\sigma_n^{1-t}&\le 
Q_t\nw\bz\bar\rho_n\|\sigma_n\jz
\le 
\sum_{\rho\in\N(n)}Q_t\nw\bz\rho^{\otimes n}\|\sigma^{\otimes n}\jz\nonumber\\
&\le 
|\N(n)|\sup_{\rho\in\N}Q_t\nw\bz\rho^{\otimes n}\|\sigma^{\otimes n}\jz\nonumber\\
&=
|\N(n)|\sup_{\rho\in\N}\bz Q_t\nw\bz\rho\|\sigma\jz\jz^n\nonumber\\
&=
|\N(n)|e^{n\psi\nw(t)},\label{Stein proof1}
\end{align}
where the first inequality is due to Lemma \ref{Lemma:old-new}, the second inequality is due to  \eqref{concavity 
complement}, the third inequality is obvious, the succeeding identity follows from the definition \eqref{Q def}, and the last 
identity is due to the definition of $\psi\nw$.
Since \eqref{ena upper bound} holds for every $t\in(0,1]$, together with \eqref{Stein proof1} it yields
$e_n(a)\le |\N(n)|e^{-n\vfi\nw(a)}$. Hence we have $\Tr\sigma_nS_{n,a}\le e_n(a)\le |\N(n)|e^{-n\vfi\nw(a)}$, 
proving \eqref{finiten bounds}. Similarly, $\Tr\bar\rho_n(I-S_{n,a})\le e^{na}e_n(a)$ yields
\begin{align}
\sup_{\rho\in \N(n)}\Tr\rho^{\otimes n}(I-S_{n,a})
&\le 
\Tr\bar\rho_n(I-S_{n,a})\nonumber\\
&\le 
e^{na}|\N(n)|e^{-n\vfi\nw(a)}\nonumber\\
&=
|\N(n)|e^{-n\hat\vfi\nw(a)}.\label{type I upper bound}
\end{align}
The submultiplicativity of the trace-norm on tensor products yields that 
$\sup_{\rho\in\N}\Tr\rho^{\otimes n}(I-S_{n,a})
\le 
\sup_{\rho\in\N(n)}\Tr\rho^{\otimes n}(I-S_{n,a})+n\delta(\N(n)))$. 
Combined with \eqref{type I upper bound}, this yields \eqref{finiten bounds1}.

The inequalities in \eqref{NP upper1}--\eqref{NP upper2} are obvious from the choice of $\delta_n$.
\end{proof}

\begin{Lemma}\label{Lemma:positive rate}
We have $\vfi\nw(a)\ge a$, and for every $a<D_1(\N\|\sigma)$, we have $\hat\vfi\nw(a)>0$.
\end{Lemma}
\begin{proof}
Note that for any $t\in(0,1)$, $a(t-1)-\psi\nw(t)=(t-1)[a-\inf_{\rho\in\N}\rsrn{\rho}{\sigma}{t}]$. 
By Lemma \ref{Lemma:dist limit},
$\lim_{t\nearrow 1}\inf_{\rho\in\N}\rsrn{\rho}{\sigma}{t}=D_1(\N\|\sigma)$.
Thus, for any $a<D_1(\N\|\sigma)$, there exists a $t_a\in(0,1)$ such that 
$a-\inf_{\rho\in\N}\rsrn{\rho}{\sigma}{t_a}<0$, and hence
$0<(t_a-1)[a-\inf_{\rho\in\N}\rsrn{\rho}{\sigma}{t_a}]\le \hat\vfi\nw(a)$. 
Finally, note that assumption \eqref{support condition} yields that $\psi\nw(1)=0$, and hence
$\vfi\nw(a)\ge a-\psi\nw(1)=a$.
\end{proof}

\begin{thm}\label{thm:Stein2}
The direct rate is lower bounded by $D_1(\N\|\sigma)$, i.e.,
\begin{equation}\label{composite rate}
R(\{\N\notimes\}_{n\in\bN}\|\{\sigma^{\otimes n}\}_{n\in\bN})\ge D_1(\N\|\sigma).
\end{equation}
\end{thm}
\begin{proof}
The proposition is trivial when $D_1(\N\|\sigma)=0$, and hence for the rest we assume $D_1(\N\|\sigma)>0$. By Lemma \ref{Lemma:positive rate}, 
for every $0<a<D_1(\N\|\sigma)$ we can find 
$0<\kappa<\vfi\nw(a)/\Delta$, so that \eqref{NP upper1}--\eqref{NP upper2} hold. Since we can take $\kappa$ arbitrarily small, and $a$ arbitrarily close to $D_1(\N\|\sigma)$, we see that 
any rate below $\sup_{0<a<D_1(\N\|\sigma)}\vfi\nw(a)$ is achievable. By Lemma \ref{Lemma:positive rate}, 
$\sup_{0<a<D_1(\N\|\sigma)}\vfi\nw(a)\ge \sup_{0<a<D_1(\N\|\sigma)}a=D_1(\N\|\sigma)$, proving the assertion.
\end{proof}

The strong converse for the simple i.i.d.~case \cite{ON2} yields immediately the strong converse for the composite i.i.d.~case. 
We include a proof for completeness.
\begin{thm}\label{thm:sc}
If $\limsup_{n\to+\infty}\frac{1}{n}\log\Tr\sigma^{\otimes n}T_n\le-r$ for some sequence of tests
$T_n,\,n\in\bN$, then
\begin{align}\label{sc}
\limsup_{n\to+\infty}\frac{1}{n}\log\inf_{\rho\in\N}\Tr\rho^{\otimes n}T_n
\le 
\inf_{t>1}\frac{t-1}{t}\left[-r+\inf_{\rho\in\N}\rsrn{\rho}{\sigma}{t}\right].
\end{align}
If $r>D_1(\N\|\sigma)$ then the RHS of \eqref{sc} is strictly negative, and hence the worst-case success probability 
$\inf_{\rho\in\N}\Tr\rho^{\otimes n}T_n$ goes to zero exponentially fast. As a consequence, \eqref{composite rate} holds as an equality.
\end{thm}
\begin{proof}
Following \cite{Nagaoka2} (see also \cite{MO}), we can use the monotonicity of the R\'enyi divergences 
under measurements for $\alpha>1$ \cite{FL,MO,Renyi_new,WWY} to obtain that for any sequence of tests $T_n,\,n\in\bN$,
any $\rho\in\N$, and any $t>1$,
\begin{align*}
&Q_t\nw(\rho^{\otimes n}\|\sigma^{\otimes n})\\
&\ge 
Q_t\nw\bz \left\{\Tr\rho^{\otimes n}T_n,\Tr\rho^{\otimes n}(I_n-T_n)\right\}\right.\|\\
&\s\ds\ds\ds\ds\left.
\left\{\Tr\sigma^{\otimes n}T_n,\Tr\sigma^{\otimes n}(I_n-T_n)\right\}\jz\\
&\ge 
\bz\Tr\rho^{\otimes n}T_n\jz^{t}\bz \Tr\sigma^{\otimes n}T_n\jz^{1-t},
\end{align*}
which yields
\begin{align*}
\frac{1}{n}\log\Tr\rho^{\otimes n}T_n
\le 
\frac{t-1}{t}\left[\frac{1}{n}\log\Tr\sigma^{\otimes n}T_n+\rsrn{\rho}{\sigma}{t}\right].
\end{align*}
Taking first the infimum in $\rho\in\N$, and then the limsup in $n$, we obtain \eqref{sc}.

Since $\inf_{t>1}\inf_{\rho\in\N}\rsrn{\rho}{\sigma}{t}=
\inf_{\rho\in\N}\inf_{t>1}\rsrn{\rho}{\sigma}{t}=D_1(\N\|\sigma)$, we see that if 
$r>D_1(\N\|\sigma)$ then there exists a $t>1$ such that $-r+\inf_{t>1}\inf_{\rho\in\N}\rsrn{\rho}{\sigma}{t}<0$, and hence the RHS of \eqref{sc} is strictly negative.
The rest of the statements follow immediately.
\end{proof}

\begin{rem}
Theorem \ref{thm:Stein2} shows the existence of a sequence of tests such that the type II error probability decays exponentially fast with rate $D_1(\N\|\sigma)$, while the
type I error probability goes to zero. Note that for this statement, it is enough to choose $\delta_n$ polynomially decaying; e.g.~$\delta_n:=1/n^2$ does the job, 
and we get an improved exponent for the type II error, $\limsup_{n\to+\infty}\frac{1}{n}\log\beta_n(S_{n,a})\le -\vfi\nw(a)$.

Theorem \ref{thm:Stein} yields more detailed information in the sense that it shows that 
for any rate $r$ below the optimal rate $D_1(\N\|\sigma)$, there exists a sequence of tests along which the type II error decays with the given rate $r$, while the type I error also 
decays exponentially fast; moreover, \eqref{NP upper1} provides a lower bound on the rate of the type I error. Note that if $\N$ is finite then the approximation process can be omitted, and we obtain the bounds 
\begin{align*}
\limsup_{n\to+\infty}\frac{1}{n}\log\alpha_n(S_{n,a})&\le -\hat\vfi\nw(a),\\
\limsup_{n\to+\infty}\frac{1}{n}\log\beta_n(S_{n,a})&\le -\vfi\nw(a).
\end{align*}
These bounds are not optimal; indeed, in the simple i.i.d.~case the quantum Hoeffding bound theorem \cite{ANSzV,Hayashi,HMO2,Nagaoka2} shows that the above 
inequalities become equalities with $\vfi\nw$ and $\hat\vfi\nw$ replaced with $\vfi\old(a):=\sup_{0<t\le 1}\{at-\log Q_t\old(\rho\|\sigma\},\,\hat\vfi\old(a):=\vfi\old(a)-a$,
and if $\rho$ and $\sigma$ don't commute then $\vfi\old(a)>\vfi\nw(a)$ and $\hat\vfi\old(a)>\hat\vfi\nw(a)$ for any $0<a<D_1(\rho\|\sigma)$, according to \cite{Hiai ALT}. On the other hand, the RHS of \eqref{sc} is known to give the exact strong converse exponent in the simple i.i.d.~case \cite{MO}.
\end{rem}
\medskip

The above arguments can also be used to obtain bounds on the direct rate in the case of states with arbitrary correlations. In this case, however, it may not be possible to find a 
suitable approximation procedure, and hence we restrict our attention to the case of finite null-hypothesis. 
Thus, for every $n\in\bN$, our alternative hypothesis $H_1$ is given by some state $\sigma_n\in\S(\hil_n)$, where $\hil_n$ is some finite-dimensional Hilbert space, and 
the null-hypothesis $H_0$ is given by
$\N_n=\{\rho_{1,n},\ldots,\rho_{r,n}\}\subset\S(\hil_n)$, where $r\in\bN$ is some fixed number.
We assume that $\supp\rho_{i,n}\subseteq\supp\sigma_n$ for every $i$ and $n$.

\begin{thm}\label{thm:correlated Stein}
In the above setting, we have
\begin{align}
&\limsup_{n\to+\infty}\frac{1}{n}\log\alpha_n(S_{n,a})\nonumber\\
&\le
-\sup_{0<t<1}\left\{a(t-1)-\max_{1\le i\le r}\limsup_{n\to+\infty}\frac{1}{n}\log Q_t\nw(\rho_{i,n}\|\sigma_n)\right\},\label{correlated Stein1}\\
&\limsup_{n\to+\infty}\frac{1}{n}\log\beta_n(S_{n,a})\nonumber\\
&\le
-\sup_{0<t<1}\left\{at-\max_{1\le i\le r}\limsup_{n\to+\infty}\frac{1}{n}\log Q_t\nw(\rho_{i,n}\|\sigma_n)\right\}
\nonumber\\
&\le -a,\label{correlated Stein3}
\end{align}
where $S_{n,a}:=\left\{e^{-na}\sum_i\rho_{i,n}-\sigma_n>0\right\}$. As a consequence, the direct rate is lower bounded as
\begin{align}\label{correlated Stein4}
R(\{\N_n\}_{n\in\bN}\|\{\sigma_n\}_{n\in\bN})&\ge
\sup_{0<t<1}\min_{1\le i\le r}\liminf_{n\to+\infty}\frac{1}{n} D_t\nw(\rho_{i,n}\|\sigma_n).
\end{align}
If $\limsup_{n\to+\infty}\frac{1}{n}\log\dim\hil_n<+\infty$ then we also have
\begin{align}\label{correlated Stein5}
R(\{\N_n\}_{n\in\bN}\|\{\sigma_n\}_{n\in\bN})&\ge
\min_i\derleft\psi_i\old(1),
\end{align}
where $\derleft$ stands for the left derivative, and $\psi_i\old(t):=\limsup_{n\to+\infty}\frac{1}{n}\log Q_t\old(\rho_{i,n}\|\sigma_n)$.
\end{thm}
\begin{proof}
The same argument as in Theorem \ref{thm:Stein} yields \eqref{correlated Stein1} and  \eqref{correlated Stein3}, from which \eqref{correlated Stein4}
follows immediately.
Assume now that $L:=\limsup_{n\to+\infty}\frac{1}{n}\log\dim\hil_n<+\infty$.
By Lemma \ref{Lemma:old-new}, we have
\begin{align}\label{correlated Stein2}
\limsup_{n\to+\infty}\frac{1}{n}\log Q_t\nw(\rho_{i,n}\|\sigma_n)\le t\psi_i\old(t)+(t-1)^2 L.
\end{align}
Note that $\psi_i\old(t)$ is the pointwise limsup of convex functions, and hence it is convex, too. Moreover, the support condition 
$\supp\rho_{i,n}\subseteq\supp\sigma_n$ implies $\psi_i\old(1)=0$. Hence, we have
$\lim_{t\nearrow 1}\frac{t}{t-1}\psi_i\old(t)=\derleft\psi_i\old(1)$. 
Combining this with \eqref{correlated Stein1} and  \eqref{correlated Stein2}, we see that $\limsup_{n\to+\infty}\frac{1}{n}\log\alpha_n(S_{n,a})<0$ for all
$a<\min_i\derleft\psi_i\old(1)$. Taking into account \eqref{correlated Stein3}, we get \eqref{correlated Stein5}.
\end{proof}

\begin{rem}\label{rem:correlated Stein}
Note that under suitable regularity, we have 
$\displaystyle{\derleft\psi_i\old(1)=\lim_{n\to+\infty}\frac{1}{n}\rsr{\rho_{i,n}}{\sigma_n}{1}}$, and hence
\begin{align}\label{correlated Stein7}
R(\{\N_n\}_{n\in\bN}\|\{\sigma_n\}_{n\in\bN})&\ge
\min_i\lim_{n\to+\infty}\frac{1}{n}\rsr{\rho_{i,n}}{\sigma_n}{1}.
\end{align}
This is clearly satisfied in the i.i.d.~case, and we recover \eqref{composite rate}. There are also various important special cases of correlated states where the above holds. 
This is the case, for instance, if all the $\rho_{i,n}$ and $\sigma_n$ are $n$-site restrictions of gauge-invariant quasi-free states on a fermionic or bosonic chain (the latter type of states are also called Gaussian states). In this case $\lim_{n\to+\infty}\frac{1}{n}\rsr{\rho_{i,n}}{\sigma_n}{1}$ can be expressed by an explicit formula in terms of the symbols of the states; see \cite{MHOF,M} for details. Another class of states where the above holds is when 
$\rho_{i,n}$ and $\sigma_n$ are group-invariant restrictions of i.i.d.~states on a spin chain \cite{HMH}. In this case one can use the same approximation procedure as in the i.i.d.~case,
and hence \eqref{correlated Stein7} holds for $\N_n:=\{\rho_{i,n}:\,i\in\I\}$, where $\I$ is an arbitrary (not necessearily finite) index set.
\end{rem}

\medskip

Finally, we show that the above considerations for the composite null-hypothesis yield the direct rate also for the \ki{averaged i.i.d.} case. In this setting
we have a probability measure $\mu$ on the Borel sets of $\S(\hil)$ such that 
$\bar\rho_n:=\int_{\S(\hil)}\rho^{\otimes n}\,d\mu$ is well-defined for every $n\in\bN$.
The null-hypothesis is given by the sequence $\N_n=\{\bar\rho_n\},\,n\in\bN$, and the alternative hypothesis is given by the sequence $\sigma^{\otimes n},\,n\in\bN$, as in the composite i.i.d.~case.
Note that in this case the null-hypotheses is simple, i.e., $\N_n$ consists of one single element, but it is not i.i.d.
Let 
\begin{align*}
D^*:=\sup\Big\{&\inf_{\rho\in\S(\hil)\setminus H}\rsr{\rho}{\sigma}{1}: \\
&\,H\subset\S(\hil)
\text{ Borel set with }\mu(H)=0 \Big\},
\end{align*}
which is essentially the relative entropy distance of $\supp\mu$ from $\sigma$, modulo subsets of zero measure.
Assume that $D^*>0$, since otherwise \eqref{averaged rate} holds trivially.
For every $0<a<D^*$, there exists a subset $\N(a)$ such that $a<\rsr{\N(a)}{\sigma}{1}\le D^*$ and 
$\mu(\S(\hil)\setminus \N(a))=0$. Applying Theorem \ref{thm:Stein} to 
the composite i.i.d.~problem with null-hypothesis $\N(a)$, we get the existence of a sequence of tests $T_n,\,n\in\bN$, such that 
\begin{align*}
&\limsup_{n\to+\infty}\frac{1}{n}\log\Tr\sigma^{\otimes n}T_n\le -a,\\
&\limsup_{n\to+\infty}\frac{1}{n}\log\Tr\bar\rho_n(I-T_n)\\
&\le 
\limsup_{n\to+\infty}\frac{1}{n}\log\sup_{\rho\in\N(a)}\Tr\rho^{\otimes n}(I-T_n)<0.
\end{align*}
Hence, the direct rate for the averaged i.i.d.~problem is lower bounded by $D^*$, i.e.,
\begin{equation}\label{averaged rate}
R(\{\bar\rho_n\}_{n\in\bN}\|\{\sigma^{\otimes n}\}_{n\in\bN})\ge D^*.
\end{equation}

\subsection{Universal state compression}
\label{sec:comp}

Consider a sequence of finite-dimensional Hilbert spaces $\hil_n,\,n\in\bN$, and for each $n$, let $\N_n\subset\S(\hil_n)$ be a set of states. 
An \ki{asymptotic compression scheme} is a sequence $(\C_n,\D_n),\,n\in\bN$, where $\C_n:\,\B(\hil^{\otimes n})\to\B(\kil_n)$ is the compression map, and 
$\D_n:\,\B(\kil_n)\to\B(\hil^{\otimes n})$ is the decompression. We use two different measures for the fidelity of $(\C_n,\D_n)$, defined as 
\begin{align*}
F(\C_n,\D_n)&:=\inf_{\rho_n\in\N_n}F_e(\rho_n,\D_n\circ\C_n),\\
\hat F(\C_n,\D_n)&:=\inf_{\rho_n\in\N_n}F(\rho_n,(\D_n\circ\C_n)\rho_n),
\end{align*}
where $F$ stands for the fidelity, and $F_e$ for the the entanglement fidelity (see Section \ref{sec:notations}).
Due to the monotonicity of the fidelity under partial trace, we have $F(\C_n,\D_n)\le\hat F(\C_n,\D_n)$.
Let $\left[\C_n(\N_n)\right]$ be the projection onto the subspace generated by the supports of $\C_n(\rho_n),\,\rho_n\in\N_n$.
We say that a compression rate $R$ is achievable if there exists 
an asymptotic compression scheme $(\C_n,\D_n),\,n\in\bN$, such that 
\begin{align*}
\lim_{n\to+\infty}F(\C_n,\D_n)&=1,\ds\ds\ds\text{and}\\
\limsup_{n\to+\infty}\frac{1}{n}\log\Tr \left[\C_n(\N_n)\right]&\le R. 
\end{align*}
The smallest achievable compression rate is the \ki{optimal compression rate} $R(\{\N_n\}_{n\in\bN})$.

We say that the compression problem is i.i.d.~if $\hil_n=\hil^{\otimes n}$ and $\N_n=\N^{(\otimes n)}:=\{\rho^{\otimes n}:\,\rho\in\N\}$ for every $n\in\bN$, 
where $\hil=\hil_1$, and $\N\subset\S(\hil)$.
It was shown in \cite{Schumacher} (see also \cite{JS}) that in the simple i.i.d.~case, projecting the state onto its entropy-typical subspace yields the entropy as an 
achievable coding rate, which 
is also optimal. In Section 10.3 of \cite{Hayashibook}, Neyman-Pearson type projections were used instead of the typical projections, and exponential bounds were obtained for 
the error probability 
for suboptimal coding rates. An extension of the typical projection technique was used in \cite{JHHH} to obtain universal state compression, i.e., it was shown that 
for any $s>0$, there exists a coding scheme of rate $s$ that is asymptotically error-free for any state of entropy less than $s$. Theorem \ref{thm:universal compression} below shows that 
the use of Neyman-Pearson projections can also be extended to obtain universal state compression. Since Theorem \ref{thm:universal compression} is essentially a special case of Theorems
\ref{thm:Stein} and \ref{thm:sc} with the choice $\sigma:=I$, we omit the proof. The only part that does not follow immediately from Theorems
\ref{thm:Stein} and \ref{thm:sc} is relating the fidelity to the success probability of the generalized state discrimination problem; this, however, is standard and we refer the interested reader to Section 12.2.2 in \cite{NC}.
\smallskip 

 Let $\psi\old(t)=\psi\nw(t)$, $\vfi\old(a)=\vfi\nw(a)$ and $\hat\vfi\old(a)=\hat\vfi\nw(a)$ be defined as in \eqref{psi}--\eqref{phi}, with $\sigma:=I$. The above equalities hold because $\rho$ and $\sigma=I$ commute for any $\rho$, and hence
$Q_t\nw(\rho\|\sigma)=Q_t\old(\rho\|\sigma)=\Tr\rho^t$.

\begin{thm}\label{thm:universal compression}
In the i.i.d.~case, for every $\kappa>0$, $a\in\bR$, and $n\in\bN$, let $\delta_n:=e^{-n\kappa}$, let 
$\N_{\delta_n}\subset\N_n$ be a subset as in Lemma \ref{Lemma:state approximation}, and let $S_{n,a}:=\left\{e^{-na}\sum_{\rho\in\N_{\delta_n}}\rho^{\otimes n}-I_n>0\right\}$. Define
\begin{align*}
\C_n(.)&:=S_{n,a}(.)S_{n,a}+\pr{x_n}\Tr(.)(I-S_{n,a}),\\
\D_n&:=\id,
\end{align*}
where $x_n$ is an arbitrary unit vector in the range of $S_{n,a}$.
For every $a\in\bR$ and $\kappa>0$, we have
\begin{align}
\limsup_{n\to+\infty}\frac{1}{n}\log(1-F(\C_n,\D_n))&\le -\min\{\kappa,\hat\vfi(a)-\kappa \Delta\},\\
\limsup_{n\to+\infty}\frac{1}{n}\log\Tr \left[\C_n(\N_n)\right]&\le -\vfi(a)+\kappa \Delta.
\end{align}

On the other hand, for any coding scheme $(\C_n,\D_n),\,n\in\bN$, we have 
\begin{align*}
&\limsup_{n\to+\infty}\frac{1}{n}\log \hat F(\C_n,\D_n)\\
&\le
\inf_{t>1}\frac{t-1}{t}\left[\limsup_{n\to+\infty}\frac{1}{n}\log\Tr \left[\C_n(\N_n)\right]-\sup_{\rho\in\N}S_t(\rho)\right].
\end{align*}
where $S_t(\rho):=\frac{1}{1-t}\log\Tr\rho^t$ is the R\'enyi entropy of $\rho$ with parameter $t$.
\end{thm}

\begin{cor}\label{cor:universal compression}
The optimal compression rate is equal to the maximum entropy, i.e.,
\begin{align*}
R(\{\N^{(\otimes n)}_{n\in\bN}\})=\sup_{\rho\in\N}S(\rho).
\end{align*}
\end{cor}

\begin{rem}
We recover the result of \cite{JHHH} by choosing $\N:=\{\rho\in\S(\hil):\,S(\rho)\le s\}$.
\end{rem}

\begin{rem}
Theorem \ref{thm:universal compression} and Corollary \ref{cor:universal compression} can be extended to correlated states and averaged states the same way as the analogous 
results for state discrimination in Section \ref{sec:Stein}. Since these extensions are trivial, we omit the details.
\end{rem}

\begin{rem}
The simple i.i.d.~state compression problem can also be formulated in an ensemble setting, which is in closer resemblance with the usual formulation of classical 
source coding. In that formulation,
a discrete i.i.d.~quantum information source is specified by a finite set $\{\rho_x\}_{x\in\X}\subset\S(\hil)$ of states and a probability distribution $p$ on $\X$. Invoking the source 
$n$ times, we obtain a state $\rho_{\vecc{x}}:=\rho_{x_1}\otimes\ldots\otimes\rho_{x_n}$ with probability $p(\vecc{x}):=p(x_1)\cdot\ldots\cdot p(x_n)$. 
The fidelity of a compression-decompression pair $(\C_n,\D_n)$ is then defined as
$F(\C_n,\D_n):=\sum_{x\in\X}p(x)F_e\bz\rho_x,\D_n\circ\C_n\jz$.
In the classical case the signals $\rho_x$ can be identified with a system of orthogonal rank $1$ projections,
$\C_n$ and $\D_n$ are classical stochastic maps, and $F(\C_n,\D_n)$ as defined above gives back the usual expression for the success probability.
It follows from standard properties of the fildelity that the optimal compression rate, under the constraint that 
$F(\C_n,\D_n)$ goes to $1$ asymptotically, only depends on the average state 
$\rho(p):=\sum_xp(x)\rho_x$, and is equal to $S(\rho(p))$. Theorem \ref{thm:universal compression} and Corollary \ref{cor:universal compression}
thus also provide the optimal compression rate and exponential bounds on the error and success probabilities in the ensemble formulation, for
multiple quantum sources.
\end{rem}

\subsection{Classical capacity of compound channels}
\label{sec:capacity}

Recall that by a channel $W$ we mean a map $W:\,\X\to\S(\hil)$, where $\X$ is some input alphabet (which can be an arbitrary non-empty set) and $\hil$ is a finite-dimensional Hilbert space. 
For a channel $W:\,\X\to\S(\hil)$, we define its \ki{$n$-th i.i.d.~extension} $W^{\otimes n}$ 
as the channel $W^{\otimes n}:\,\X^n\to\S(\hil^{\otimes n})$, defined as
\begin{equation}
W^{\otimes n}(x_1,\ldots,x_n):=W(x_1)\otimes\ldots \otimes W(x_n),
\end{equation}
$x_1,\ldots,x_n\in\X$.

It is obvious from the explicit formula \eqref{explicit alpha-Holevo} for $\chi_{\alpha}\old(W,p)$ that
\begin{equation}\label{chi additivity}
\chi_{\alpha}\old(W^{\otimes n},p^{\otimes n})=n\chi_{\alpha}\old(W,p),\ds\ds\ds n\in\bN,
\end{equation}
where $p^{\otimes n}\in\P_f(\X^n)$ is the $n$-th i.i.d.~extension of $p$, defined as
$p^{\otimes n}(x_1,\ldots,x_n):=p(x_1)\cdot\ldots\cdot p(x_n)$, $x_1,\ldots,x_n\in\X$.
It follows from \cite[Theorem 11]{Beigi} that the same additivity property holds for $\hol{\alpha}\nw$, i.e.,
\begin{equation}\label{chi additivity2}
\chi_{\alpha}\nw(W^{\otimes n},p^{\otimes n})=n\chi_{\alpha}\nw(W,p),\ds\ds\ds n\in\bN.
\end{equation}
Note, however, that while the proof of \eqref{chi additivity} is almost trivial, the proof of 
\eqref{chi additivity2} is mathematically very involved.

\begin{rem}
Note that in our definition of a channel, we didn't make any assumption on the cardinality of the input alphabet $\X$, nor did we require any further mathematical properties from $W$, apart from being a function to $\S(\hil)$. 
The usual notion of a quantum channel is a special case of this definition, where $\X$ is the state space of some Hilbert space and $W$ is a completely positive trace-preserving convex map. In this case, however, our definition of the 
i.i.d.~extensions are more restrictive 
than the usual definition of the tensor powers of a quantum channel. Indeed, our definition corresponds to the notion of quantum channels with product state encoding. Hence, our definition of the classical capacity below
corresponds to the classical capacity of quantum channels with product state encoding. 
\end{rem}

Let $W_i:\,\X\to\S(\hil),\,i\in\I$, be a set of channels with the same input alphabet $\X$ and the same output Hilbert space $\hil$, where $\I$ is any index set.
A \ki{code} $\C=(\C_e,\C_d)$ for $\{W_i\}_{i\in\I}$ consists of an encoding
$\C_e:\{1,\ldots,M\}\to\X$ and a decoding $\C_d:\{1,\ldots,M\}\to\B(\hil)_+$, where
$\{\C_d(1),\ldots,\C_d(M)\}$ is a POVM 
on $\hil$, and $M\in\bN$ is the size of the code, which we will denote by 
$|\C|$. The elements of $\mathrm{ran}\,\C_e$ are called the \ki{codewords} of $\C$.
The worst-case average error probability of a code $\C$ is
\begin{align*}
p_e\bz\{W_i\}_{i\in\I},\C\jz&:=
\sup_{i\in\I}\frac{1}{|\C|}\sum_{k=1}^{|\C|}\Tr W_i(\C_e(k))(I-\C_d(k)).
\end{align*}
When the set $\{W_i\}_{i\in\I}$ contains only one single channel $W$, we will use the simpler notation $p_e(W,\C)$ for the error probability.

Consider now a sequence $\W:=\{\W_n\}_{n\in\bN}$, where each $\W_n$ is a set of channels with input alphabet $\X^n$ and output space $\hil^{\otimes n}$. 
The \ki{classical capacity} $C(\W)$ of $\W$ is the largest number $R$ such that there exists a sequence of codes $C^{(n)}=\bz C^{(n)}_e,C^{(n)}_d\jz$ with
\begin{align*}
\lim_{n\to+\infty}p_e(\W_n,\C_n)=0\ds\ds\ds\text{and}\ds\ds\ds
\liminf_{n\to+\infty}\frac{1}{n}\log|\C_n|\ge R.
\end{align*}

We say that $\W$ is simple i.i.d.~if $\W_n$ consists of one single element $W^{\otimes n}$ for every $n\in\bN$ with some fixed channel $W$. 
In this case we denote the capacity by $C(W)$.
The Holevo-Schumacher-Westmoreland theorem \cite{Holevo,SW} tells that in this case
\begin{equation}\label{HSW}
C(W)\ge\chi(W)=\sup_{p\in\P_f(\X)}\chi(W,p),
\end{equation}
where $\chi(W,p)$ is the Holevo quantity \eqref{Hol def}, and $\chi(W)$ is the Holevo capacity \eqref{Holevo cap} of the channel.
It is easy to see that 
\eqref{HSW} actually holds as an equality, i.e., no sequence of codes with a rate above
$\sup_{p\in\P_f(\X)}\chi(W,p)$ can have an asymptotic error equal to zero;
this is called the weak converse to the channel coding theorem, while the strong converse theorem 
\cite{ON,Winter} says that such sequences of codes always have an asymptotic error equal to $1$.

Here we will consider two generalizations of the simple i.i.d.~case:
In the \ki{compound i.i.d.} case $\W_n=\{W_i^{\otimes n}\}_{i\in\I}$ for some fixed channels
$W_i:\,\X\to\S(\hil)$. In the \ki{averaged i.i.d.} case
$\W_n$ consists of the single element $\ol W_n:=\sum_{i\in\I}\gamma_iW_i^{\otimes n}$, where
$\I$ is finite, and
$\gamma$ is a probability distribution on $\I$.
The capacity of finite averaged channels has been shown to be equal to 
$\sup_{p\in\P_f(\X)}\min_i\chi(W_i,p)$ in \cite{DD}, 
and the same formula for the capacity of a finite compound channel follows from it in a straightforward way.
The protocol used in \cite{DD} to show the achievability was to use a certain fraction of the communication rounds to guess which 
channel the parties are actually using, and then code for that channel in the remaining rounds.
These results were generalized to arbitray index sets $\I$
in \cite{BB}, using a different approach. The starting point in \cite{BB} was the following random coding theorem from \cite{HN}
(for the exact form below, see \cite{MD}).

\begin{Lemma}\label{thm:error bound}
Let $W:\,\X\to\S(\hil)$ be a channel.
For any $M\in\bN$, and any $p\in\P_f(\X)$, there exists a code $\C$ with codewords in $\supp p$ such that $|\C|=M$ and
\begin{equation*}
p_{e}(W,\C)\le\kappa(c,\alpha) M^{1-\alpha}\Tr \what W(p)^{\alpha}(\hat p\otimes W(p))^{1-\alpha}
\end{equation*}
for every $\alpha\in(0,1)$ and every $c>0$, 
where $\kappa(c,\alpha):=(1+c)^{\alpha}(2+c+1/c)^{1-\alpha}$.
\end{Lemma}

Applying the general properties of the R\'enyi divergences, established in Section \ref{sec:Renyi}, together 
with the single-shot coding theorem of Lemma \ref{thm:error bound}, we get a very simple proof of the achievability part 
of the coding theorems in \cite{DD} and \cite{BB}. Since our primary interest is the applicability of the 
new R\'enyi divergences $D_{\alpha}\nw$ to achievability proofs, we will not consider the converse
parts.
%
The key step of our approach is the following extension of Lemma \ref{thm:error bound} to multiple channels.

\begin{Lemma}\label{thm:capacity}
Let $W_i:\,\X\to\S(\hil),\,i\in \I$, be a set of channels, where $\I$ is a finite index set.
For every $R\ge 0$, every $n\in\bN$, and every $p\in\P_f(\X)$, there exists a code
$\C_n$ with codewords in $\supp p^{\otimes n}$, such that for every $\alpha\in(0,1)$,
\begin{align}
&|\C_n|\ge \exp(nR),\ds\ds\ds\ds\ds\ds\text{and}\nonumber\\
&p_e\bz\{W_i^{\otimes n}\}_{i\in\I},\C_n\jz\nonumber\\
&\ds\le
8|\I|^2\exp\Big[ n(\alpha-1)\Big(\alpha\min_i\chi_{\alpha}\old(W_i,p)-R \nonumber\\
&\ds\ds\ds\ds\ds\ds\ds\ds\ds\ds\ds\ds\ds\ds\ds\ds-(\alpha-1)\log\dim(\hil)\Big)\Big].
\label{channel error bound}
\end{align}
\end{Lemma}
\begin{proof}
Let $M_n:=\lceil \exp(nR)\rceil,\,n\in\bN$ and $\gamma_i:=1/|\I|,\,i\in\I$.
Applying Lemma \ref{thm:error bound} to 
$\ol W_n=\sum_{i\in\I}\gamma_iW_i^{\otimes n}$, $M_n$ and $p^{\otimes n}$, we get the existence 
of a code $\C_n$ with codewords in $\supp p^{\otimes n}$ and $|\C_n|=M_n$, such that
\begin{align}
&p_{e}(\ol W_n,\C_n)\le \nonumber\\
&\ds8M_n^{1-\alpha}
Q_{\alpha}\old\bz \sum_{i\in\I}\gamma_i\what W_i^{\otimes n}(p^{\otimes n})\Big\|\hat p^{\otimes n}\otimes \ol  W_n(p^{\otimes n})\jz\label{n-shot bound}
\end{align}
for every $\alpha\in(0,1)$. Here we chose $c=1$, and used the upper bound $\kappa(1,\alpha)\le 8$.
We can further upper bound the RHS above as
\begin{align}
&Q_{\alpha}\old\bz \sum_{i\in\I}\gamma_i\what W_i^{\otimes n}(p^{\otimes n})\Big\|\hat p^{\otimes n}\otimes \ol  W_n(p^{\otimes n})\jz\nonumber\\
&\ds\le
Q_{\alpha}\nw\bz \sum_{i\in\I}\gamma_i\what W_i^{\otimes n}(p^{\otimes n})
 \Big\|\hat p^{\otimes n}\otimes \ol  W_n(p^{\otimes n})\jz\label{ineqq1}\\
&\ds\le
\sum_{i\in\I}\gamma_i^{\alpha}
Q_{\alpha}\nw\bz \what W_i^{\otimes n}(p^{\otimes n})\big\|\hat p^{\otimes n}\otimes \ol  W_n(p^{\otimes n})\jz\label{ineqq2}\\
&\ds\le
\sum_{i\in\I}\gamma_i^{\alpha}\sup_{\sigma\in\S(\hil^{\otimes n})}
Q_{\alpha}\nw\bz \what W_i^{\otimes n}(p^{\otimes n})\big\|\hat p^{\otimes n}\otimes \sigma\jz
\label{ineqq3}\\
&\ds\le
\sum_{i\in\I}\gamma_i^{\alpha}\sup_{\sigma\in\S(\hil^{\otimes n})}
Q_{\alpha}\old\bz \what W_i^{\otimes n}(p^{\otimes n})\big\|\hat p^{\otimes n}\otimes \sigma\jz^{\alpha}\nonumber\\
&\ds\ds\ds\ds\ds\ds\ds \cdot(\dim\hil^{\otimes n})^{(\alpha-1)^2}\label{ineqq4}\\
&\ds=
\sum_{i\in\I}\gamma_i^{\alpha}\exp\bz \alpha(\alpha-1)\chi_{\alpha}\old(W_i^{\otimes n},p^{\otimes n}\jz(\dim\hil)^{n(\alpha-1)^2}\label{ineqq5}\\
&\ds=
\sum_{i\in\I}\gamma_i^{\alpha}\exp\bz n\alpha(\alpha-1)\chi_{\alpha}\old(W_i,p)\jz(\dim\hil)^{n(\alpha-1)^2}
\label{ineqq6},\\
&\ds\le
|\I|\exp\bz n\alpha(\alpha-1)\min_{i\in\I}\chi_{\alpha}\old(W_i,p)\jz(\dim\hil)^{n(\alpha-1)^2}
\label{ineqq7}
\end{align}
where \eqref{ineqq1} is due to the first inequality in \eqref{old-new bounds}, \eqref{ineqq2} is due to 
the second inequality in \eqref{concavity complement}, \eqref{ineqq3} is trivial, 
\eqref{ineqq4} follows from \eqref{old-new bounds10},
and \eqref{ineqq6} is due to \eqref{chi additivity}.
Note that 
\begin{align}
p_{e}(\ol W_n,\C_n)&=
\frac{1}{|\I|}\sum_{i\in\I}p_e(W_i^{\otimes n},\C_n)
\ge
\frac{1}{|\I|}\sup_{i\in\I}p_e(W_i^{\otimes n},\C_n).\label{ineqq8}
\end{align}
Combining \eqref{n-shot bound}, \eqref{ineqq7}, and \eqref{ineqq8}, we get
\eqref{channel error bound}.
\end{proof}

\begin{rem}
We could have chosen a slightly different path above, and instead of switching back to the $Q_{\alpha}$ quantities in 
\eqref{ineqq4}, use directly the additivity \eqref{chi additivity2} of $\hol{\alpha}\nw$ to obtain a bound similar to the one in 
\eqref{ineqq6}, but in terms of the $\hol{\alpha}\nw$ quantities. Since the $\hol{\alpha}\nw$ quantities also yield the Holevo quantity in the limit $\alpha\to 1$, this bound would be equally useful for Theorem \ref{thm:compound cap}. The reason 
that we followed the above path instead is to use as little technically involved ingredients in the proof as possible, and 
the proof of the the additivity of the $\hol{\alpha}\old$ quantities is considerably simpler than for the 
$\hol{\alpha}\nw$ quantities.
\end{rem}

The above Lemma yields almost immediately the coding theorem for compound channels:
\begin{thm}\label{thm:compound cap}
Let $W_i:\,\X\to\S(\hil),\,i\in \I$, be a set of channels, where $\I$ is an arbitrary index set.
Then
\begin{align}\label{finite capacity}
C\bz\{W_i^{\otimes n}:\,i\in\I\}_{n\in\bN}\jz
\ge \sup_{p\in\P_f(\X)}\inf_i\chi(W_i,p).
\end{align}
\end{thm}
\begin{proof}
We assume that $\sup_{p\in\P_f(\X)}\inf_i\chi(W_i,p)>0$, since otherwise the assertion is trivial.
Let $p\in\P_f(\X)$ be such that $\inf_i\chi(W_i,p)>0$, and for every $i\in\I$, let 
$W_{p,i}:\,\supp p\to\S(\hil)$ be the restriction of the channel $W_i$ to $\supp p$. 
Let $V$ be the vector space of functions from $\X$ to $\B(\hil)$, equipped with the norm
$\norm{V}:=\sup_{x\in\supp p}\norm{V(x)}_1$, and let 
$\Delta$ denote the real dimension of $V$.
Let $\kappa>0$, and for every $n\in\bN$, let $\I(n)$ be a finite index set such that $|\I(n)|\le (1+2e^{n\kappa})^{\Delta}$ and
$\delta_n:=\sup_{i\in\I}\inf_{j\in\I(n)}\norm{W_{p,i}-W_{p,j}}\le e^{-n\kappa}$. The existence of such index sets is guaranteed by Lemma \ref{Lemma:state approximation}. 

Let $R$ be such that $0<R<\inf_i\chi(W,p)$, and
for every $n\in\bN$, let $\C_n$ be a code as in Lemma \ref{thm:capacity}, with $\I(n)$ in place of $\I$,
and $\{W_{p,i}\}_{i\in\I(n)}$ in place of $\{W_i\}_{i\in\I}$.
Since the codewords of $\C_n$ are in $\supp p^{\otimes n}$, we have
\begin{align*}
p_e\bz\{W_{p,i}^{\otimes n}\}_{i\in\I(n)},\C_n\jz=
p_e\bz\{W_i^{\otimes n}\}_{i\in\I(n)},\C_n\jz,
\end{align*}
and it is easy to see that 
\begin{align*}
p_e\bz\{W_i^{\otimes n}\}_{i\in\I(n)},\C_n\jz
\ge
p_e\bz\{W_i^{\otimes n}\}_{i\in\I},\C_n\jz-n\delta_n.
\end{align*}
Hence, by Lemma \ref{thm:capacity} we have
\begin{align*}
&p_e\bz\{W_i^{\otimes n}\}_{i\in\I},\C_n\jz\\
&\ds\le
8|\I(n)|^2\exp\Big[ n(\alpha-1)\Big(\alpha\inf_{i\in\I}\chi_{\alpha}\old(W_i,p)-R\\
&\ds\ds\ds-(\alpha-1)\log\dim(\hil)\Big)\Big]
+ne^{-n\kappa},
\end{align*}
where we also used that $(\alpha-1)\min_{i\in\I(n)}\chi_{\alpha}\old(W_{p,i},p)
=(\alpha-1)\min_{i\in\I(n)}\chi_{\alpha}\old(W_{i},p)
\le(\alpha-1)\inf_{i\in\I}\chi_{\alpha}\old(W_{i},p)$.

By Lemma \ref{Lemma:inf Hol limit}, there exists an
$\alpha\in(0,1)$ such that 
$\nu:=\alpha\inf_{i\in\I}\chi_{\alpha}\old(W_i,p)-R-(\alpha-1)\log\dim(\hil)>0$. Choosing then
$\kappa$ such that $2\kappa \Delta/(1-\alpha)<\nu$, we see that
the error probability goes to zero exponentially fast, while the rate is at least $R$.

This shows that $C\bz\{W_i^{\otimes n}:\,i\in\I\}_{n\in\bN}\jz\ge \inf_i\chi(W_i,p)$, and taking the supremum 
over all $p\in\P_f(\X)$ yields the assertion.
\end{proof}

Theorem \ref{thm:compound cap} yields immediately the following lower bound on the capacity of finite averaged channels,
which is the achievability part of the coding theorem in \cite{DD}:
\begin{cor}\label{cor:finite capacity lower bounds}
Let $W_i:\,\X\to\S(\hil),\,i\in \I$, be a set of channels, where $\I$ is an arbitrary index set, and let 
$\gamma$ be a finitely supported probability distribution on $\I$.
Then
\begin{align*}
C\bz\left\{\sum\nolimits_i\gamma(i) W_i^{\otimes n}\right\}_{n\in\bN}\jz
&=
C\bz\{W_i^{\otimes n}:\,i\in\supp \gamma\}_{n\in\bN}\jz\\
&\ge \sup_{p\in\P_f(\X)}\min_{i\in\supp \gamma}\chi(W_i,p).
\end{align*}
\end{cor}

\section*{Acknowledgments}

The author is grateful to Professor Fumio Hiai and Nilanjana Datta for discussions, and to an anonymous
referee for helpful suggestions regarding the presentation of the paper.
This research was supported by a Marie Curie International Incoming Fellowship within 
the 7th European Community Framework Programme.
The author also acknowledges support by the European Research Council 
(Advanced Grant ``IRQUAT''), by the Spanish MINECO Project No. FIS2013-40627-P, and by the Generalitat de Catalunya CIRIT Project No. 2014 SGR 966.
Part of this work was done when the author was a Marie Curie 
research fellow at the School of Mathematics, University of Bristol.

\begin{IEEEbiographynophoto}
{Mil\'an Mosonyi} Received his PhD in Physics from the Catholic University of Leuven in 2005. He 
joined the Department of Analysis at the Budapest University of Technology and Economics as an assistant professor in 2005, and
he has been an associate professor there since 2012. Currently he is on a 2-year leave at the F\'{\i}sica Te\`{o}rica: Informaci\'{o} i Fenomens Qu\`{a}ntics,
Universitat Aut\`{o}noma de Barcelona, as a postdoctoral research fellow. His main research interests are quantum Shannon theory and quantum statistics.
\end{IEEEbiographynophoto}

\end{document}